\documentclass[12pt,reqno]{amsart}

\usepackage{amsmath,amsfonts,amsbsy,amsthm}
\usepackage{amssymb}
\usepackage{eucal}
\usepackage{hyperref}
\usepackage{dsfont}
\usepackage{xy}
\xyoption{matrix} \xyoption{arrow} \xyoption{arc} \xyoption{color}

\usepackage{enumerate}
\usepackage{enumitem}
\setlist{leftmargin=*}

\usepackage{tikz}
\usetikzlibrary{decorations.pathmorphing,shapes,arrows}

\usepackage{verbatim} 
\numberwithin{equation}{section}

\makeatletter
\newtheoremstyle{corsivo}
   {\medskipamount}{\medskipamount}%
   {\itshape}{}%
   {\bfseries}{}%
   { }
   {\thmname{#1}\thmnumber{\@ifnotempty{#1}{ }\@upn{#2}}%
    \thmnote{ {\bfseries\boldmath(#3)}}.}%
\makeatother

\theoremstyle{corsivo}
\newtheorem{theorem}{Theorem}[section]
\newtheorem{lemma}[theorem]{Lemma}

\newtheorem{proposition}[theorem]{Proposition}
\newtheorem{definition}[theorem]{Definition}
\newtheorem{assumption}[theorem]{Assumption}

\makeatletter
\newtheoremstyle{dritto}
   {\medskipamount}{\medskipamount}%
   {\rmfamily}{}%
   {\bfseries}{}%
   { }
   {\thmname{#1}\thmnumber{\@ifnotempty{#1}{ }\@upn{#2}}%
    \thmnote{ {\bfseries\boldmath(#3)}}.}%
\makeatother

\theoremstyle{dritto}

\newtheorem{remark}[theorem]{Remark}

\newcommand{\sub}[1]{_{\mathrm{#1}}}

\newcommand{\Id}{\mathds{1}}  %%%%% Identity operator on H
\newcommand{\iu}{\mathrm{i}}   %%%%% Imaginary unit
\newcommand{\di}{\mathrm{d}}

\newcommand{\up}{\uparrow}
\newcommand{\down}{\downarrow}

\newcommand{\J}{\mathbf{J}} 
\newcommand{\X}{\mathbf{X}}

\newcommand{\N}{\mathbb{N}}
\newcommand{\Z}{\mathbb{Z}}

\newcommand{\R}{\mathbb{R}}
\newcommand{\C}{\mathbb{C}}

\newcommand{\BH}{\mathcal{B}(\mathcal{H})}
\newcommand{\BHi}{\mathcal{B}(\mathcal{H}_\mathrm{disc})}
\newcommand{\Do}{\mathcal{D}}

\newcommand{\Hi}{\mathcal{H}_{\mathrm{disc}}}

\newcommand{\B}{\mathcal{B}}

\newcommand{\scal}[2]{\left\langle #1 , #2 \right\rangle}                
\newcommand{\inner}[2]{\left\langle #1 , #2 \right\rangle}     
\newcommand{\norm}[1]{\left\| #1 \right\|}

\newcommand{\ket}[1]{\left| #1 \right\rangle}

\newcommand{\set}[1]{ \left\{  #1 \right\}} 

\DeclareMathOperator{\Tr}{Tr}         %  Hilbert space trace
\DeclareMathOperator{\tr}{tr}           %  Lie algebra or matrix trace

\DeclareMathOperator{\re}{Re} \DeclareMathOperator{\im}{Im}

\DeclareMathOperator{\Ran}{Ran}

\newcommand{\ie}{{\sl i.\,e.\ }}   %%% id est
\newcommand{\eg}{{\sl e.\,g.\ }} %%%  exemplum gratiae
\newcommand{\virg}[1]{``#1''}
\newcommand{\crucial}[1]{{\it \textbf{#1}}}
\newcommand{\half}{\mbox{\footnotesize $\frac{1}{2}$}}

\renewcommand{\(}{\left(}
\renewcommand{\)}{\right)}

\newcommand{\E}{{\mathrm{e}}}

\newcommand{\V}[1]{\mathbf{#1}}
\newcommand{\abs}[1]{\left\lvert#1\right\rvert}
\newcommand{\pv}{{\mathrm{p.v.}}}
\newcommand{\pvTr}{{\mathrm{pvTr}}} 
\newcommand{\TC}{\mathcal{B}_1(\mathcal{H})}
\newcommand{\TCi}{\mathcal{B}_1(\mathcal{H}_\mathrm{disc})}

\newcommand{\jpvTr}{{j\mathrm{\text{-}pvTr}}}  %%% j-directional principal value trace
\newcommand{\tuv}{\textsc{tuv}}    %%%%%% trace per unit volume (TUV) in the text

\makeatletter
\newcommand{\srcsize}{\normalfont\@setfontsize{\srcsize}{6pt}{6pt}}
\makeatother

\newcommand{\G}{{\mathcal{G}}}
\newcommand{\adj}{{\mathrm{adj}}}
\newcommand{\SQ}{{\mathcal{Q}}}
\newcommand{\Br}{\mathbb{B}}

\newcommand{\ee}{{\mathrm e}}
\newcommand{\ii}{{\mathrm i}}

\let\oldfootnote\footnote
\renewcommand{\footnote}[1]{\oldfootnote{\  #1}}

\title[]{Spin Conductance and Spin Conductivity \\[1mm] in Topological Insulators: \\[1mm] 
Analysis of Kubo-like terms}
\author[G. Marcelli, G. Panati, C. Tauber]{Giovanna Marcelli, Gianluca Panati \and Cl\'ement Tauber}
\date{January 7, 2018. Version submitted to \textsl{arXiv.org}}

\usepackage{hyperref}

\begin{document}

\begin{abstract} 
We investigate spin transport in $2$-dimensional insulators, with the long-term goal of establishing whether any of the transport coefficients corresponds to the Fu-Kane-Mele index which characterizes
$2d$ time-reversal-symmetric topological insulators. 

Inspired by the Kubo theory of charge transport, and by using a proper definition of the 
spin current operator \cite{ShiZhangXiaoNiu}, we define the Kubo-like spin conductance $G_K^{s_z}$  and spin conductivity $\sigma_K^{s_z}$.  
We prove that for any gapped, periodic, near-sighted discrete Hamiltonian, the above quantities are mathematically well-defined and the equality $G_K^{s_z} = \sigma_K^{s_z}$ holds true. 
Moreover, we argue that the physically relevant condition to obtain the equality above is the vanishing of the mesoscopic average of the spin-torque response, which holds true under our hypotheses on the Hamiltonian operator.  This vanishing condition might be relevant in view of further extensions of the result, \eg to ergodic random discrete Hamiltonians or to Schr\"odinger operators on the continuum. 
A central role in the proof is played by the trace per unit volume and by two generalizations of the trace, 
the \emph{principal value trace} and it directional version. %This machinery might be of independent interest. 

\end{abstract}

\maketitle

\vspace{-12mm}
\tableofcontents

%\newpage
\goodbreak

%%%%%%%%%%%%%%%%%%%%%%%%%%%%%%%%%%%%%%%%%

%\vspace{5mm}

\section{Introduction}
\label{Sec:Intro}

The last few decades witnessed an increasing interest, among solid state physicists,  for physical phenomena having a topological origin. This interest traces back to the milestone paper by Thouless, Kohmoto, Nightingale and den Nijs on the Quantum Hall Effect (QHE) \cite{TKNN},  %%%
includes the pioneering work of Haldane on Chern insulators \cite{Haldane88} %%%%
and the seminal papers by Fu, Kane and Mele concerning the Quantum Spin Hall Effect (QSHE) \cite{KaneMele2005,KaneMele_graphene, FuKa, FuKaneMele} %%%
up to the most recent developments in the flourishing field of topological insulators 
\cite{Ando,HasanKane}.

As it is well-known, in the QHE a topological invariant (Chern number) is related to an observable quantity, 
the transverse charge conductance or Hall conductance. 
By analogy, in the context of the QSHE for $2$-dimensional time-reversal-symmetric insulators, 
one would like to connect -- if possible -- the relevant topological invariant (Fu-Kane-Mele index) 
to a macroscopically observable quantity. 
The natural candidates are spin conductance and spin conductivity, whose proper definition has been debated, and whose equivalence has not been yet established.

The first crucial point is to characterize the operator corresponding to the \emph{spin current density}.
In the last few years, an intense debate about the correct expression of the latter took place, 
but a general consensus was not reached 
\cite{ShiZhangXiaoNiu, cinesi, Schulz-BaldesCMP, Murakami, AnLiuLinLiu, BrayNussinov, SunXieWang}. %% 
Among the candidates, one may include:%%%%%%%%
\footnote{We use Hartree atomic units, so that the reduced Planck constant $\hbar$, the squared electron charge $e^2$ and 
the electron mass $m\sub{e}$ are dimensionless and equal to $1$. In particular, the quantum of charge conductivity in the QHE 
is $\frac{e^2}{h} = \frac{1}{2\pi}$. 
}  %%%% End footnote   
\renewcommand{\labelenumi}{{\rm(\roman{enumi})}}
\begin{enumerate} 
\item  the naive guess
$$
\J \sub{naive} = \iu [H,\X] \, S_z,
$$
where $H$  is the Hamiltonian operator of the system, $\X=(X_1,X_2)$ is the position operator 
and $S_z$ represents the $z$-component of the spin;
%\footnote{Throughout the paper, $S_z := \dfrac{\hbar}{2} \sigma_z$ where $\sigma_z$ is the third component of Pauli matrix vector $\boldsymbol{\sigma}=(\sigma_x,\sigma_y, \sigma_z)$.};
\item its symmetrized version, namely
$$
\J\sub{sym} =  \half \(\J\sub{naive} + \J\sub{naive}^* \) = \half \( \iu [H,\X] \, S_z + \iu  S_z \, [H,\X] \),
$$
which has the advantage of providing a {self-adjoint} operator; 
\item  last but not least, the alternative provided by the \virg{proper} spin current
\begin{equation} 
\label{J_proper}
\mathbf{J}\sub{prop} = \iu [H, \mathbf{X} S_z], 
\end{equation}
proposed by \cite{ShiZhangXiaoNiu}, which is also self-adjoint. 
\end{enumerate} 
Whenever $[H,S_z]=0$ (spin-commuting case), the three above definitions agree, 
while they differ in general. Notice that spin conservation is often violated in  
topological insulators, as it happens \eg in the paradigmatic model proposed by Kane and Mele
\cite{KaneMele2005, KaneMele_graphene}, reviewed in Appendix \ref{Sec:Kane-Mele}. 
Hence, it is of prominent importance to understand which choice best models the physics. \newline  
The choice (iii) has the advantage to provide an operator associated to a sourceless continuity equation for the associated density and to Onsager relations \cite{{ShiZhangXiaoNiu, cinesi}}. On the other hand, $\J\sub{sym}$ provides a \emph{periodic} (or covariant, when ergodic randomness is added) operator, while - as early remarked by Schulz-Baldes -  the latter property fails to hold for $\J\sub{prop}$, which \virg{ {\it leads to technical difficulties, but also questions the physical relevance }} of the operator $\J\sub{prop}$ \cite{Schulz-BaldesCMP}. 

In this paper, we are inspired by the following simple but new observation: even if  $\J\sub{prop}$ is not periodic, it satisfies a peculiar commutation relation with the lattice translations $\set{T_{\mathbf{p}}}_{\mathbf{p}\in \Z^d}$ whenever the Hamiltonian operator is periodic. Namely,  
%\begin{equation} \label{J-translated}
%[T_{\mathbf{p}},  \J\sub{prop}] =   \V{p}\, T_{\V{p}}[H,S_z]. 
%\end{equation} 
\begin{equation} \label{J-translated}
  T_{\mathbf{p}} \, \J\sub{prop} \, T_{\mathbf{p}} ^{-1}  =  \J\sub{prop} -  \V{p} \,\, T_{\V{p}} \, \iu [H,S_z] \, T_{\mathbf{p}}^{-1}
  \qquad \forall \V{p} \in \Z^d.  
\end{equation}
Hence, whenever the \emph{spin torque} $\iu [H,S_z]$ averages to zero on the mesoscopic scale, \eg because 
{$\tau\( \ii [H,S_z] \rho(t)\)=0$} where $\tau(\,\cdot\,)$ is the trace per unit volume (see Definition \ref{defn:trace per unit vol}) {and $\rho(t)$ is the density matrix describing the state of the system}, the operator $\J\sub{prop}$ is 
\crucial{\virg{mesoscopically periodic}}, in the sense that its commutator with the lattice translations vanishes on the mesoscopic scale. 

\medskip 

A second crucial question is whether the relevant observable quantity related to the Fu-Kane-Mele (FKM) index is the spin conductance, or the spin conductivity, or some other transport coefficient, if any.  
We recall that the transverse (resp.\ direct)   spin conductance is defined, experimentally, as the ratio between the spin current intensity and the electric potential drop measured in orthogonal (resp.\ parallel) directions, hence as the ratio of two extensive observable quantities. On the contrary, the  transverse (resp.\ direct)  spin conductivity is the ratio between the spin current and the strength of the electric field measured in orthogonal  (resp.\ parallel)  directions, and as such is the ratio of two intensive quantities.  In the case of \emph{charge transport} in $2$-dimensional systems, the equality of charge conductance and conductivity holds true \cite{AvronSeilerSimon}, under suitable technical hypotheses, at least within the Linear Response Approximation (LRA) \cite{AizenmanGraf, Graf review, AizenmanWarzel}. %% BoucletGerminet2005, ElgartSchlein}.
In the case of \emph{spin transport} the situation is instead radically different and, unless $[H, S_z]=0$, it is not obvious  {\it a priori} whether the equality between spin conductance and spin conductivity holds true or not.   

\medskip 

%%%%%%%%%%%  Our work  %%%%%%%%%%%%%%%%%%%%%%%%%%%%

Our analysis encompasses several steps. 
As a first step, we reconsider the spin transport starting from the first principles of Quantum Mechanics. 
This analysis, performed in two related papers \cite{MMPTe, MPTa} by a space- and a time-adiabatic approach, respectively, 
shows that spin conductivity and conductance, 
defined by using the operator $\J\sub{prop}$ (whose lack of periodicity is harmless on the mesoscopic scale, as remarked above), 
contain additional terms with respect to what suggested by the analogy with the 
Kubo theory of charge transport. The physical relevance of the additional terms is at the moment unclear, and deserves further investigations by both numerical  and analytical methods. 

As a second step, in this paper we investigate the \crucial{Kubo-like terms}. Explicitly, they are the following:

\begin{enumerate}[label={\rm (\alph*)},ref={\rm (\alph*)}]
%%%%%%  Spin Conductivity  %%%%%%%%%%%%%%%%%%%%%%%%%  
\item \label{def:sigma_K}  the \crucial{Kubo-like spin conductivity} is defined as 
\begin{equation} \label{Eq:sigma_K}
{\sigma}_K^{s_z}:=\tau (\Sigma_K^{s_z}) \quad\text{ with }\quad \Sigma_K^{s_z}:=\ii P \, \big[ [P, X_1 S_z], [P,X_2] \big] \,P
\end{equation}
where $P$ is the Fermi projector up to energy $\mu \in \R$, which is supposed to be in a spectral gap,
and $\tau(\,\cdot\,)$ is the trace per unit volume (\tuv). 
The fact that $\tau (\Sigma_K^{s_z})$ is well-defined and finite will be part of our results. \\  
%%%%%%  Spin Conductance  %%%%%%%%%%%%%%%%%%%%%%%%%  
\item \label{def:G_K}  the \crucial{Kubo-like spin conductance} is defined as
\begin{align} \label{Eq:G_K}
G_K^{s_z}(\Lambda_1,\Lambda_2):=&1\text{-}\pvTr\big(\G_K^{s_z}(\Lambda_1,\Lambda_2)\big) \\[2mm]
& \text{ with } \G_K^{s_z}(\Lambda_1,\Lambda_2):=  \ii P \, \big[ [P, \Lambda_1 S_z], [P,\Lambda_2] \big] \, P
\nonumber
\end{align}
where $\Lambda_j$ is a convenient  switch function in direction $j \in \set{1,2}$, as in Definition \ref{defn:switch funct}. 
The fact that the operator  $\G_K^{s_z}(\Lambda_1,\Lambda_2)$ is not trace class, forces us to introduce a suitable trace-like linear functional, denoted by  $1\text{-}\pvTr$ and baptized  \emph{directional principal value trace} in 
direction $j=1$ in Definition\,\ref{defn:jpv tr}, which generalizes the trace. 
 \end{enumerate}

The first new result of our paper is that, when focusing on the Kubo-like terms \eqref{Eq:sigma_K} and \eqref{Eq:G_K}, 
spin conductance and conductivity are equal provided that $\tau({\mathcal{T}_{s_z}}) =0$, where 
the \crucial{spin torque-response} operator is defined by
\begin{equation}
\label{eqn:defn spin torque}
\mathcal{T}_{s_z}:=\iu P\, \big[ [P,S_z], [P,X_2] \big] \, P.
\end{equation}
Physically,  $\tau({\mathcal{T}_{s_z}})$ represents -- within LRA -- the response of the system, in terms of spin torque $\iu [H,S_z]$, to a uniform electric field in direction $2$. 

The second new result is that, for any \emph{periodic} and near-sighted Hamiltonian (compare Assumption \ref{ass:H}), 
condition $\tau (\mathcal{T}_{s_z}) =0$ automatically holds true, so that we conclude that 
$G_K^{s_z} = {\sigma}_K^{s_z}$.  In particular, under these assumptions the spin conductance is independent of the switch functions involved in its definition. 
The precise results, which for technical reasons are proved in the setting of discrete Hamiltonians, 
are stated in Theorem \ref{thm:main1} and Theorem \ref{thm:main2},  while the crucial
observation mentioned after \eqref{J-translated} reflects  {in equations \eqref{eqn:Sigma notper: per+oddper}, \eqref{eqn:not per odd op} and \eqref{eqn:KG12s=Ksigma12s general proof} in the proofs}.

Notice that our results do not assume the smallness of $[H,S_z]$, hence they go beyond the regime of spin quasi-conservation considered in previous papers \cite{Schulz-BaldesCMP, Prodan1}. %\NB{completare le citazioni} 

%%%%% The mathematical machinery %%%%%%%%%%%%%%%%%%%%%

\medskip

To prove our results we need to set up a suitable mathematical machinery, involving some 
trace-like linear functionals, as the \emph{principal value trace} (Definition 1.5) and the 
\emph{$j$-directional principal value trace} (Definition 1.6). We also prove some
relevant properties of the trace per unit volume (Definition 1.7). \newline
As it is well-known, in an infinite dimensional Hilbert space one has in general $\Tr([A,B]) \neq 0$, 
since the cyclicity of the trace holds true only under special conditions, \eg  if $AB$ and $BA$ are trace class 
and both $A$ and $B$ are bounded operators (see \cite{Simon} and references therein). 
Similar subtleties appear when considering the trace-like functionals mentioned above.   
{It is noteworthy that}, many physically relevant quantities appear as the trace or \tuv \ of exact commutators.
For example, as noticed in \cite{AvronSeilerSimon} the Kubo charge conductance $\sigma_K^e$  
for a Quantum Hall system can be rewritten as 
$$
\sigma_K^e = \tau \( [PX_1P, P X_2 P] \)  %\frac{1}{2\pi} \,\,
$$
where $P$ is the spectral projector up to the Fermi energy.
Hence, the mentioned mathematical subtleties are not an abstract academic issue, but are deeply intertwined with the physics of 
quantum transport. For this reason, we devote two sections to the analysis of the properties of the mentioned 
trace-like functionals (Sections \ref{sect:pv tr jpv tr tau prop} and \ref{Sec:Trace near-sighted}), also considering that part of this machinery might be of independent interest. In this analysis, we greatly benefited by the previous 
work on charge transport in Quantum Hall systems, including in particular \cite{AizenmanGraf, AvronSeilerSimon, BoucletGerminet2005, ElgartGrafSchenker, ElgartSchlein}.
The mathematical setting and the main results  are discussed in Section \ref{sec:math setting and main result}, 
while Section \ref{sect:main results} is devoted to the proofs. 

%%%%%% Perspectives  %%%%%%%%%%%%%%%%%%%

%\medskip

Our work provides a mathematical consistent expression for the Kubo-like terms of spin conductivity and conductance, and  some sufficient conditions which imply their equality.  Moreover, our work puts on solid mathematical grounds the proposal to use $\J\sub{prop}$ as the self-adjoint operator corresponding to spin current density, 
circumventing the criticism related to its failure to be periodic.  
These results pave the way to further developments in the mathematical theory of time-reversal-symmetric topological insulators, a very active field of research in Solid State Physics and, more recently, in Mathematical Physics 
\cite{Prodan1, Prodan3, Frohlich, AvilaSchulz-Baldes12,  Schulz-BaldesCMP, Schulz-Baldes13, GrafPorta, FiMoPa, PaMo, Lyon, Lyon15,  DeNittisGomi, KatsuraKoma, CorneanMonacoTeufel, MoTa, Gawedzki}. In particular, our results might contribute to solve one of the most challenging problem in the field, namely to find a quantitative relation between an observable quantity and the relevant topological invariant,  the Fu-Kane-Mele index. 

Although our results are restricted to periodic discrete models for technical reasons,
the general strategy of the proof might presumably be applied also to ergodic random models for spin transport, 
\ie to  the natural generalization of the models considered in the context of charge transport 
in Quantum Hall systems \cite {AizenmanGraf, AizenmanWarzel, ElgartGrafSchenker, BoucletGerminet2005}. 

%\medskip

\noindent \textbf{Acknowledgements.} We are indebted to Gian Michele Graf for sharing with us his insight into the mathematics of the QHE on the occasion of the Winter School \virg{{\it The Mathematics of Topological Insulators in Naples }},  
organized in the framework of the Cond-Math project (\href{http://www.cond-math.it/}{\texttt{http://www.cond-math.it/}}), 
and for pointing out to us some relevant references. 
We are grateful to Domenico Monaco and Stefan Teufel for many useful discussions, and to Massimo Moscolari for a careful reading of the manuscript.

\goodbreak

%%%%%%%%%%%%%%%%%%%%%%%%%%%%%%%%%%%%%%%%%%%%%%%%%%%%%%%%%%%%%

\section{Setting and main results}
\label{sec:math setting and main result}

%%%% Hilbert space  %%%%%%%%%%%%%%%%%%%%%%%%%%%

We consider independent electrons moving in a discrete set $\mathcal{C}\subset \R^2$, 
which is supposed to be a \crucial{periodic crystal}, \ie it is equipped with a free action of a Bravais lattice $\Gamma\simeq\Z^2$. 
In view of the latter action, after a choice of a periodicity cell, one decomposes $\mathcal{C} \simeq \Z^2 \times \{ \nu_1,\ldots, \nu_N \}$, where the second factor corresponds to the \virg{points inside the chosen periodicity cell} (see Appendix~\ref{Sec:Kane-Mele} for the specific case of the honeycomb structure and the Kane-Mele model).

Taking spin into account, the Hilbert space of the system is  $\mathcal{H}\sub{phys}=\ell^{2}(\mathcal{C})  \otimes  \C^2$ 
which, in view of the above procedure, %(often called \emph{dimerization}), 
is identified with
\begin{equation}
\label{eqn:iso for Hdisc}
\Hi=\ell^{2}(\Z^2) \otimes \C^N \otimes  \C^2.
\end{equation} 

%%%% Operators  %%%%%%%%%%%%%%%%%%%%%%%%%%%

Any bounded operator $A$ acting on $\Hi$ is identified with a collection of matrices 
$\set{A_{\V{n},\V{m}}}_{{\V{n},\V{m}} \in \Z^2} \subset \mathrm{End}(\C^N\otimes \C^2)$.
Indeed, by choosing any orthonormal basis $\{ {e}_j \}_{ j\in \{1,\ldots,N \}}$ for $\C^N$ 
and any orthonormal basis $\set{\phi_s}_{s \in \set{ \uparrow, \downarrow}}$ for spin, 
%setting $\{  \uparrow, \downarrow\}:=  
%\left\lbrace \begin{pmatrix} 
%1  \\ 
%0 
%\end{pmatrix} 
%, 
%\begin{pmatrix} 
%0 \\ 
%1
%\end{pmatrix} \right\rbrace$, 
a bounded operator $A$ is characterized by the matrices  
%$$
%\ket{\V{n},j,s}:= \ket{ \delta_\V{n}\otimes \V{e}_j\otimes s }\in\Hi\quad\text{for $\V{n}\in\Z^2, j\in\{1,\ldots,N \}$ and $s \in \{\uparrow,\downarrow\} $,}
%$$ 
%which corresponds to the state localized at site $\V{n}\in\Z^2$ with crystal point $\V{e}_j$ and spin $s\in \{\uparrow,\downarrow\}$.
$$
A_{\V{n},\V{m}}:= \set{ \inner{ \delta_\V{n}\otimes {e}_j\otimes \phi_s}{A  ( \delta_\V{m}\otimes {e}_k \otimes \phi_r)} 
}_{\{j,k\in \{1,\ldots,N \},\,s,r\in\{\uparrow, \downarrow \}\}} \in \mathrm{End}(\C^N\otimes \C^2) 
$$
for all $\V{n},\V{m}\in\Z^2$, where $\delta_\V{n}$ is defined as usual by $(\delta_\V{n})_{\V{m}} = \delta_{\V{n}, \V{m}}$. 
We denote by $\abs{A_{\V{m},\V{n}}}$ the corresponding matrix norm, while the operator norm on the full Hilbert space $\Hi$ is denoted by $\norm{A}$.  

%%%% Hamiltonian  %%%%%%%%%%%%%%%%%%%%%%%%%%%	
	
\begin{definition}
\label{defn:op near-sighted}
A bounded operator $A$ acting on $\Hi$ is called \textbf{near-sighted}\footnote{The term near-sighted was proposed by the Nobel Laureate Walter Kohn \cite{Kohn,ProdanKohn}, in a slightly different context. For electrons in crystals, \virg{it describes the fact that [...] local electronic properties [...] depend significantly on the effective external potential only at nearby points.} The term \textbf{short range operator} is often equivalently used in the  literature, as well as \textbf{local operator}. The latter use, however overlaps with the standard meaning of the word \virg{local} in the theory of operators, so we avoid it.} if and only if there exist constants $C,\zeta>0$ such that 
$$
\abs{A_{\V{m},\V{n}}} \leq C \ee^{-\frac{1}{\zeta} \norm{\V{m}-\V{n}}_1}\quad\forall \V{m},\V{n}\in\Z^2,
$$
where $\norm{\V{n}}_1:=\sum_{j=1}^2\abs{n_j}$. The constant $\zeta$ is called the \textbf{range} of $A$.
\end{definition}

\goodbreak

\begin{assumption}
\label{ass:H}
The Hamiltonian operator $H$ is a bounded self-adjoint operator acting on $\Hi$. Further, we assume that the operator $H$ 
\begin{enumerate}[label={\rm (H$_\arabic*$)},ref={\rm (H.\arabic*)}]
\item  \label{item:H near-sighted} is near-sighted with range $\zeta_H$;
\item  \label{item:H periodic} is periodic, namely $H_{\V{m},\V{n}}=H_{\V{m}-\V{p}, \V{n}-\V{p}}$ for all $\V{m},\V{n},\V{p}\in\Z^2$;
\item  \label{item:H spectral gap} admits a spectral gap, 
            namely there exist non-empty sets $I_1,I_2 \subseteq \R$ and $a,b\in\R$, such that 
            \begin{equation*}
            \mathrm{Spectrum}(H)=I_1\cup I_2\text{  and  }\sup I_1 < a< b< \inf I_2.
            \end{equation*}
            The interval $\Delta =(a,b)$ is called \textbf{spectral gap}.
\end{enumerate}
\end{assumption}

For $\mu \in \Delta$, we denote the Fermi projection by 
\begin{equation}
\label{eqn:fermi proj}
P := \chi_{(-\infty,\mu)}(H),
\end{equation}
where {$\chi_{\Omega}$} is the characteristic function of the set {$\Omega$}. 
In Appendix~\ref{Sec:Kane-Mele}, we show that the Hamiltonian $H\sub{KM}$ of the Kane-Mele model, which is often considered the paradigmatic model of time-reversal-symmetric topological insulators, enjoys all the above assumptions, whenever the values of the parameters guarantee the existence of the spectral gap. Moreover, one easily sees that $[H\sub{KM}, S_z]\neq 0$.

%%%%%%%%%  MAIN %%%%%%%%%%%%%%%%%%%%%%%%%%%%%%%%%%%

\medskip

The aim of this paper is to analyze the Kubo-like terms in the spin conductivity and spin conductance, 
defined as in \eqref{Eq:sigma_K} and \eqref{Eq:G_K}, respectively.  
In our context, the position operator  $\V{X} = (X_1, X_2)$  acts in $\Hi$ as
$$
(X_j \varphi )_{\V{n}} := n_j \varphi_{\V{n}}, \quad j \in \{1,2\}, \quad \forall \varphi \in \Do(X_j). 
$$

The spin operator $S_z$ acts on $\Hi$  as  $\Id \otimes  \Id \otimes \half s_z$, where $s_z$ is the third Pauli matrix. 
In order to keep a light notation, in the following we identify any operator $A$ which acts only in one sector of $\Hi$, with the one acting in $\Hi$ with extra identity factors, and we keep the same notation $A$ (\eg $X_1 \equiv X_1 \otimes \Id_{\C^N} \otimes \Id_{\C^2}$, and so on). 

The operator $\G_K^{s_z}$ involves the notion of switch function, which we now define. 

\begin{definition}
\label{defn:switch funct}
Fix $j\in\{1,2 \}$.  A \textbf{switch function in the $j^{\mathrm{th}}$-direction} is a function $\Lambda_j\colon\Z^2\to [0,1]$ that depends only on the variable $n_j$ and satisfies
\[
\Lambda_j(n_j)=
\begin{cases}
0 & \text{if $n_j < n_-$} \\
1 & \text{if $n_j\geq n_+$}
\end{cases}
\]
for arbitrary $n_- < n_+$.
\end{definition}

%%%%% Trace-like functionals %%%%%%%%%%%%%%%%%%%%%%%%%%%%%%%

As anticipated in the introduction, many subtleties of the  quantum theory of transport arise since some relevant operators appearing in the theory are not trace class. 
%\gm{In realt\'a nel caso carica gli analoghi operatori sono tracciabili o tracciabili per unit\'a di volume, invece sono i singoli termini del commutatore che danno problemi, quindi questa frase la toglierei: \virg{
The operators $\Sigma_K^{s_z}$ and $\G_K^{s_z}$, defined in \eqref{Eq:sigma_K} and \eqref{Eq:G_K}, are not exceptional. 
To overcome this problem, one needs to define suitable trace-like linear functionals corresponding to the relevant physical quantities. The transverse spin conductivity is defined through the well-known trace per unit volume. However, for the conductance the situation is quite different and we have to introduce the notions of \crucial{principal value trace} and its \crucial{directional} version.

We make use of the norm
$$
\norm{\V{n}}_{\infty}:=\max_{j\in\set{1,2}}\abs{n_j}\quad\forall \V{n}\in\Z^2,
$$
which conveniently respect the square structure of $\Z^2$. For any $L\in 2\N+1$ and $\V{n}_0\in\Z^2$, 
we set
$$
\SQ_L(\V{n}_0):=\set{\V{n}\in\Z^2:\; \norm{\V{n}-\V{n}_0}_\infty\leq L/2}
$$
to denote the square of side $L$ centered at $\V{n}_0$. Following \cite{BoucletGerminet2005}, we restrict 
to odd integers ($L\in 2\N+1$) in order to use the convenient decomposition\footnote{The symbol $\bigsqcup$ corresponds to the disjoint union.} 
\begin{equation}
\label{eqn:squareL is L^2 unit squares}
\SQ_{L}(\V{n}_0)= \bigsqcup_{\V{n}\in \SQ_L(\V{n}_0)}\SQ_1(\V{n}).
\end{equation}
%namely the square of odd side $L$ centered at $\V{n}_0$ is exactly the disjoint union of $L^2$ unit square centered at $\V{n}\in\SQ_{L}(\V{n}_0)$. 
For the sake of better readability, we write $\SQ_L$ for $\SQ_L(\V{0})$.  
%namely the length of the shortest path between two vertices $\V{n}$ and $\V{m}$ of the graph $\Z^2$ is given by the norm $\norm{\V{n-m}}_1$. We call two vertices $\V{n}$ and $\V{m}$ in $\Z^2$ \textbf{nearest neighbors} if and only if $\norm{\V{n-m}}_1=1$, and \textbf{next-to-nearest neighbors} if and only if $\norm{\V{n-m}}_1=2$.
\newline
We denote by $\chi_L:=\chi_{\SQ_L}$, for $L\in 2\N+1$, the characteristic function of the square $\SQ_L$, 
and by $\chi_{j,L}$, for $j\in\set{1,2}$ and $L\in 2\N+1$,  the characteristic function of the stripe $\set{\V{m}\in\Z^2: \abs{m_j}\leq L/2}$.

\begin{definition}[Principal value trace]
\label{defn:pv trace}
Let $A$ be an operator acting in $\Hi$ such that%%%
\footnote{\label{fn:chiLAchiL}The condition that \virg{ $\chi_L A \chi_L$ is trace class for every $L\in 2\N+1$ } is automatically satisfied in every discrete model, as those considered in this paper, since the range of $\chi_L$ is finite-dimensional. 
We decided to state this redundant condition anyhow, since we prefer to consider the same definition for discrete and continuum models (Schr\"{o}dinger operators), as we plan to adapt the proof to the latter models in the future.  
}\  %%%%
 $\chi_L A \chi_L$ is trace class for every $L\in 2\N+1$. The principal value trace of $A$, is defined, whenever the limit exists, as
$$
\pvTr(A):=\lim_{\substack{L\to\infty\\L\in 2\N+1}} \Tr(\chi_L A\chi_L).
$$
\end{definition}

As we deal with a two-dimensional system, we can also define the notion of directional principal value trace  depending on the $j^{\mathrm{th}}$-direction, where $j\in\set{1,2}$ indicates the direction around which we localize.

\begin{definition}[Directional principal value trace]
\label{defn:jpv tr}
Fix an index $j\in\set{1,2}$. Let $A$ be an operator acting in $\Hi$ such that $\chi_{j,L} A \chi_{j,L}$ is trace class for every $L\in 2\N+1$. The $j$-directional principal value trace of $A$, is defined, whenever the limit exists, as
$$
\jpvTr(A):=\lim_{\substack{L\to\infty\\L\in 2\N+1}} \Tr(\chi_{j,L} A\chi_{j,L}).
$$
\end{definition}
We will show in  Section \ref{sect:pv tr jpv tr tau prop} that both the principal value trace and its directional version coincide with the usual trace whenever $A$ is a trace class operator. However, the new functionals work also for operators which are not trace class, in analogy with generalized integrals. 
Finally, we recall the definition of trace per unit volume (see \cite{AizenmanWarzel, BoucletGerminet2005} and references therein).

\begin{definition}[Trace per unit volume]
\label{defn:trace per unit vol}
Let $A$ be an operator acting in $\Hi$ such {that $^{\ref{fn:chiLAchiL}}$} $\chi_L A \chi_L$ is trace class for every $L\in 2\N+1$. The trace per unit volume of $A$, is defined, whenever the limit exists, as
$$
\tau(A):=\lim_{\substack{L\to\infty\\L\in 2\N+1}}\frac{1}{L^2}\Tr(\chi_L A \chi_L).
$$
\end{definition}

\noindent The fundamental properties of these three trace-like linear functionals are discussed in Section~\ref{sect:pv tr jpv tr tau prop}. 

%%%%%%%  Results  %%%%%%%%%%%%%%%%%%%%%%%%%%%%%%%%%%%

We are finally in the position to discuss the main results of the paper. We first state an auxiliary lemma. 

\begin{lemma}
\label{lem:spin torque}
Let $H$ be as in Assumption~\ref{ass:H} and $P$ be the corresponding Fermi projection, as in \eqref{eqn:fermi proj}. 
Then the spin torque-response operator $\mathcal{T}_{s_z} = \iu P\, \big[ [P,S_z], [P,X_2] \big] \, P$ is periodic and bounded. 
Moreover, $\mathcal{T}_{s_z}$ has finite trace per unit volume and it holds
$$
\tau ( \mathcal{T}_{s_z} )=\Tr(\chi_1\mathcal{T}_{s_z}\chi_1). 
$$
$\tau(\mathcal{T}_{s_z})$ is called the \textbf{mesoscopic average of spin torque-response}.
\end{lemma}

\begin{theorem}[Vanishing of spin-torque response]
\label{thm:main1}
Let $H$ be as in Assumption~\ref{ass:H} and $P$ be the corresponding Fermi projection, as in \eqref{eqn:fermi proj}.  
Then 
$$
\tau ( \mathcal{T}_{s_z} )=0. 
$$
\end{theorem}

The physical interpretation of this result is that a uniform electric field does not induce any particular spin torque excess in the sample, at least within LRA \cite{ShiZhangXiaoNiu}. The proof of it relies on the conditional cyclicity of \tuv\,which,  while false in general, holds true for a specific class of operators, as proved in Proposition \ref{prop:cycl of tau}.

\begin{theorem}
\label{thm:main2}
Let $H$ be as in Assumption~\ref{ass:H} and $P$ the corresponding Fermi projection. 
Then:
\begin{enumerate}[label=(\arabic*), ref=(\arabic*)]
\item \label{item:defn GKs}
Let $\Lambda_2$ be a fixed switch function in the $2^{\mathrm{nd}}$-direction. 
Assume that  $G_K^{s_z}(\Lambda_1,\Lambda_2)$, defined by \eqref{Eq:G_K}, 
is finite for at least a switch function $\Lambda_1$. \newline
Then $G_K^{s_z}(\Lambda_1^{\prime} ,\Lambda_2)$ 
 is finite for any of switch function $\Lambda_1^{\prime}$, and it is independent of the choice of $\Lambda_1^{\prime}$. 
 \vspace{2mm}
\item 
\label{item:defn sigmaKs} The operator $\Sigma_K^{s_z}$ satisfies
\begin{equation}
\label{eqn:Sigma notper: per+oddper}
{\(\Sigma_K^{s_z}\)}_{\V{m},\V{n}}={\(\Sigma_K^{s_z}\)}_{\V{m-p},\V{n-p}}- p_1 \, {\(\mathcal{T}_{s_z}\)}_{\V{m-p},\V{n-p}}\quad\text{for all $\V{m},\V{n},\V{p}\in\Z^2$},
\end{equation}
where $\mathcal{T}_{s_z}$ is the spin torque-response defined in \eqref{eqn:defn spin torque}. 
Moreover, 
the Kubo-like term in the transverse spin conductivity,  defined as 
${\sigma}_K^{s_z}:=\tau (\Sigma_K^{s_z})$,
is well-defined and satisfies
$$
{\sigma}_K^{s_z}=\Tr(\chi_1\Sigma_K^{s_z} \chi_1).
$$
\vspace{2mm}
\item 
\label{item:KG12s=Ksigma12s} 
Finally, the equality
\begin{equation}
\label{eqn:KG12s=Ksigma12s}
\sigma_K^{s_z} = G_K^{s_z}( \Lambda_1, \Lambda_2)
\end{equation}
holds true. In particular, $G_K^{s_z}$ is finite and independent of the choice of the switch functions $\Lambda_1,\Lambda_2$ in both directions.
\end{enumerate}
\end{theorem}

\goodbreak

\begin{remark}
Before proving the above statements, a few comments are in order.
\begin{enumerate}
\item Notice that the operator $\Sigma_K^{s_z}$ is, in general, not periodic, hence the fact that its trace per unit volume is well-defined and finite, as proved in the Theorem~\ref{thm:main2}~\ref{item:defn sigmaKs}, is not trivial.
\item The simplicity of the formula \eqref{eqn:KG12s=Ksigma12s} might obscure the physics of the problem. Indeed, during the proof, one shows that it holds true (see equation~\eqref{eqn:KG12s=Ksigma12s general proof})
\begin{align}
\label{eqn:KG12s=Ksigma12s general}
G_K^{s_z}( \Lambda_1, \Lambda_2)=\sigma_K^{s_z}+\frac{1}{2}\lim_{\substack{L\to\infty \\ L\in 2\N+1}}\sum_{\substack{m_1\in\Z\\ \abs{m_1}\leq L/2}}\tau(\mathcal{T}_{s_z}).
\end{align}
The second summand is a series of constant terms, which is either zero if $\tau(\mathcal{T}_{s_z})=0$, or $\pm \infty$ otherwise. As stated in Theorem~\ref{thm:main1}, for a gapped \emph{periodic} near-sighted Hamiltonian, one has always $\tau(\mathcal{T}_{s_z})=0$. On the other hand, we suspect that equation \eqref{eqn:KG12s=Ksigma12s general} is valid in a broader context.
\item Whenever 
\begin{equation}
\label{eqn:spincomm}
[H,S_z]=0,
\end{equation}
the spin torque-response operator vanishes, see \eqref{eqn:defn spin torque}. In this particular case, it is straightforward to prove that $G_K^{s_z}( \Lambda_1, \Lambda_2)=\sigma_K^{s_z}$, since the proof boils down to the analogous proof for charge transport (see \cite{AvronSeilerSimon} for the continuum case, and \cite{MarcelliPhD} for a recent overview of the literature). 

In view of \eqref{eqn:spincomm}, $P$ admits the decomposition induced by the $S_z$-eigenspaces, namely
\[ P = P_\up \oplus P_\down. \]
In the above, $P_\up$ and $P_\down$ are both projections on $\ell^{2}(\Z^2) \otimes \C^N$. 
In this specific case, if $H$ enjoys Assumption~\ref{ass:H} and is time-reversal symmetric, namely $\Theta H \Theta^{-1}=H$ for $\Theta = \E^{\iu \pi s_y/2} K$, where {$s_y$ is the second Pauli matrix} and $K$ is the natural complex conjugation on $\Hi$, one has that
\begin{align}
\label{eqn:KG12s=Ksigma12s general spincomm}
\sigma_K^{s_z}=\ii\tau(P \big[ [P, X_1 ], \, [P,X_2] \big]S_z P)=\frac{1}{2}\big(C_1(P_\uparrow)-C_1(P_\downarrow) \big)=C_1(P_\uparrow),
\end{align}
with 
\begin{align*}
C_1(P_s):=\frac{\iu}{2\pi}\int_{\Br}\di k\,\tr \( P_s(k)[\partial_1 P_s(k),\partial_2 P_s(k)]\)\quad\text{for $s\in\{\uparrow,\downarrow\}$},
\end{align*}
where $P_s(k)$ refers to the fiber operator at fixed crystal momentum, with respect to the modified Bloch-Floquet transform (see \eg \cite{Panati, PaMo}). 

Hence, in the spin-commuting case our result agrees with previous contributions, \eg \cite{Schulz-BaldesCMP, Schulz-Baldes13, Prodan1}, yielding that the Kubo-like spin conductivity, given by  \eqref{eqn:KG12s=Ksigma12s general spincomm}, agrees with the Spin-Chern number. Moreover, formula~\eqref{eqn:KG12s=Ksigma12s general spincomm} agrees with the Fu-Kane-Mele index modulo $2$ \cite{FuKa,KaneMele2005,Schulz-Baldes13}. 
\end{enumerate}
\end{remark}

\goodbreak

%%%%%%%%%%%%%%%%%%%%%%%%%%%%%%%%%%%%%%%
\newpage

\section{Machinery: (directional) principal value trace and trace per unit volume}
\label{sect:pv tr jpv tr tau prop}
In this Section we state and prove some fundamental properties of the trace-like functionals introduced before. 
First, we recall some facts about the trace and its conditional cyclicity.

\begin{proposition}[Conditional cyclicity of the trace {\cite[Corollary 3.8]{Simon}}]
\label{prop:cycl of trace}

Let $\mathcal{H}$ be a separable Hilbert space. If $A,B\in \mathcal{B}(\mathcal{H})$ have the property that both $AB$ and $BA$ are in the trace class ideal%%%
\footnote{{For $1\leq r<\infty$ one defines the Schatten ideals as $\mathcal{B}_r(\mathcal{H}):=\set{ A\in\mathcal{B}(\mathcal{H})\, \big| \, \abs{A}^r\in \mathcal{B}_1(\mathcal{H})}.$} 
}\  %%%%
 $\mathcal{B}_1(\mathcal{H})$ (in particular, if $A\in\mathcal{B}(\mathcal{H})$ and $B\in\mathcal{B}_1(\mathcal{H})$, or $A\in\mathcal{B}_p(\mathcal{H})$ and $B\in\mathcal{B}_q(\mathcal{H})$ where $1< p,q<\infty $ are such that $1/p+1/q=1$), then 
 $$
 \Tr(AB)=\Tr(BA).
 $$
\end{proposition}
 
 Hereafter, the  trace on the Hilbert space $\Hi$ will be denoted by $\Tr{A}$, for any trace class operator $A$, 
 while the (matrix) trace on $\C^N \otimes  \C^2 \simeq \C^{2N}$ by $\tr(\,\cdot\,)$.  
%for any operator $S$ acting on $\C^N \otimes  \C^2$. 
The following  elementary inequality will be useful. 

\begin{lemma}
\label{rem:ineq for sa op on diagonal}
Let $\mathcal{H}$ be a separable Hilbert space. If $A$ is a bounded  self-adjoint operator acting on $\mathcal{H}$, then 
\begin{equation}
\label{eqn:ineq for sa op on diagonal}
\abs{\scal{\psi}{A\psi}}\leq \scal{\psi}{\abs{A}\psi}\quad\text{for all $\psi\in\mathcal{H}$.}
\end{equation}
\end{lemma}
\begin{proof} By the Spectral Theorem,  any self-adjoint $A$ can be written as $A=A_+-A_-$, so that both $A_+$ and $A_-$ are positive operators,  $A_+A_-=0$ and $\abs{A}=A_++A_-$. 
Hence, for every $\psi\in\mathcal{H}$, one has
\begin{align*}
\abs{\scal{\psi}{A\psi}}&=\abs{\scal{\psi}{A_+-A_-\psi}}\leq \abs{\scal{\psi}{A_+\psi}}+\abs{\scal{\psi}{A_-\psi}}=\cr
&=\scal{\psi}{A_+\psi}+\scal{\psi}{A_-\psi}=\scal{\psi}{ \abs{A} \psi}.
\end{align*}
\end{proof}

\noindent Notice that the inequality~\eqref{eqn:ineq for sa op on diagonal} may be false for a bounded operator which is not self-adjoint. 

\goodbreak

%  DETAILS OF THE COUNTEREXAMPLE [do not cancel]  %%%%%%%
% For example, just consider
%\[
%A = \begin{pmatrix}
%1 & \ii \\ 1 & \ii
% \end{pmatrix}\quad\text{acting on $\C^2$}.
% \]
%
%  then we have
%\[
%\abs{A} = \begin{pmatrix}
%1 & \ii \\ -\ii & 1
% \end{pmatrix}.
%\]
%%Indeed, putting $B:=\begin{pmatrix}
%%1 & \ii \\ -\ii & 1
%% \end{pmatrix}$, we have that $B$ is positive, because its trace is $2$ and its determinant is $0$, and $B^2=A^*A$.
%Choosing $\psi=\begin{pmatrix}1\\1\end{pmatrix}$, one has $\abs{\scal{\psi}{ A \psi}}=\sqrt{8}>2= \scal{\psi}{ |A|\psi}$.
%
% END OF THE COUNTEREXAMPLE

%%%%%%  Absolute convergence of the trace series  %%%%%%%%%%%%%%%%

Chosen an orthonormal basis $\set{e_k}_{k\in\set{1,\dots,2N}}$ of $\C^N \otimes \C^2 \simeq \C^{2N}$, 
we set 
\begin{equation} \label{Def: delta basis}
\delta_{\V{n}}^{(k)}  := \delta_{\V{n}} \otimes e_{k} \in \Hi
\end{equation}
where $\delta_{\V{n}}$ is defined as usual by $(\delta_{\V{n}})_{\V{m}} = \delta_{\V{n}, \V{m}}$. If $A$ is a trace class operator, 
its trace can be computed by using the basis above, yielding 
\begin{equation} \label{Eq:Trace}
\Tr(A) %=   \sum_{\V{n}\in\Z^2}  \sum_{k\in\set{1,\dots,2N}} \inner{ \delta_{\V{n}}^{(k)}  } {A  \delta_{\V{n}}^{(k)} } 
=  \sum_{\V{n}\in\Z^2} {\tr(A_{\V{n},\V{n}})}.
\end{equation} 
{In this case,} one says that  $\Tr(A)$ is computed \virg{through the diagonal kernel}. The relevant point, recalled in the next Lemma, is that 
whenever $A$ is self-adjoint the series \eqref{Eq:Trace} is \emph{absolutely} convergent, hence the sum of the series can be obtain as the limit of the sums over sets $\Omega_n$, for any exhaustion $\Omega_n \nearrow \Z^2$.  

\begin{lemma}
\label{lem:A sa tc Ann is l1}
If $A\in\TCi$ is self-adjoint, then the function
\[
\Z^2\ni\V{n}\mapsto\tr(A_{\V{n},\V{n}})\text{ is in $\ell^1(\Z^2)$.}
\]
\end{lemma}
\begin{proof}
By the inequality~\eqref{eqn:ineq for sa op on diagonal} and the hypothesis that $A$ is trace class, one obtains
\begin{align*}
\sum_{\V{n}\in\Z^2}\abs{\tr(A_{\V{n},\V{n}})}
%%%%
&\leq \sum_{\substack{\V{n}\in\Z^2\\k\in\set{1,\dots, 2N}}}
\abs{\scal{  \delta_{\V{n}}^{(k)}  }{A \, \delta_{\V{n}}^{(k)} } }
%%%%%
\leq \sum_{\substack{\V{n}\in\Z^2\\k\in\set{1,\dots, 2N}}}
\scal{\delta_{\V{n}}^{(k)}}{\abs{A} \, \delta_{\V{n}}^{(k)}}=\cr
%%%%
&=\sum_{\V{n}\in\Z^2}\tr(\abs{A}_{\V{n},\V{n}})=\Tr(\abs{A})<\infty.
\end{align*}
%\textcolor{red}{Since $A$ is trace class, its trace is basis-independent. The diagonal kernel expression is obtained by choosing $\delta_{\V{n}}^{(k)}$.}
This completes the proof.
\end{proof}

The construction of the trace is somehow analogous to the construction of the Lebesgue integral 
\cite[Section VI.6]{ReedSimon1}. As well-known, whenever a function is Lebesgue integrable, then its principal value integral exists and it is equal to the Lebesgue integral. Similarly, the principal value trace is a {natural} extension of the trace, as stated in the following Proposition.

\begin{proposition}
\label{prop:pv tr}
If $A\in\TCi$ then $\pvTr(A)$ is well-defined and $$\pvTr(A)=\Tr(A).$$
\end{proposition}
\begin{proof}
It is sufficient to prove the claim for a self-adjoint operator $A$, since any operator $A$ can be decomposed as 
$A=\re A+\iu \im A$, and $\TC$ is closed under the adjoint operation $A \mapsto A^*$. 

%%%%% DETAILS  %%%%%%%%%%%%%%%%
%Indeed, suppose the Lemma holds true for any self-adjoint operator, then given a not self-adjoint operator $A$ we reason as follows. We decompose $A$ as the linear combination of self-adjoint operators, namely 
%$$
%A=\re A+\iu \im A,
%$$ 
%where $\re A:=\frac{A+A^*}{2}$ and $\im A:=\frac{A-A^*}{2\iu}$. As $\TC$ is closed under the adjointness, $\re A$ and $\im A$ are trace class. Thus, using the linearity of the trace and the claim for any self-adjoint operator, we have
%\[
%\Tr(A)=\Tr(\re A)+\iu\Tr(\im A)=\pvTr(\re A)+\iu\,\pvTr(\im A)=\pvTr(A).
%\]
%%%%%%%%

So, let $A$ be a self-adjoint operator. Notice that the operator $\chi_L A \chi_L$ is trace class. 
%because $\chi_L$ is a bounded operator and $A$ is trace class by hypothesis. %%%
Both the trace of $\chi_L A \chi_L$ and of $A$ can be computed through the diagonal kernel, yielding 
\begin{equation*}
\Tr(A)=\sum_{\V{n}\in\Z^2}\tr(A_{\V{n},\V{n}})=\lim_{\substack{L\to\infty\\L\in 2\N+1}}\sum_{\substack{\V{n}\in\Z^2\\ \norm{\V{n}}_\infty\leq L/2}}\tr(A_{\V{n},\V{n}})=\lim_{\substack{L\to\infty\\L\in 2\N+1}}\Tr(\chi_L A \chi_L)=\pvTr(A).
\end{equation*}
In the second equality in the last equation we have used that the function 
$\Z^2\ni\V{n}\mapsto\tr(A_{\V{n},\V{n}})$ is in $\ell^1(\Z^2)$ by Lemma~\ref{lem:A sa tc Ann is l1},
hence the series over $\Z^2$ can be computed through a particular exhaustion of $\Z^2$.
\end{proof}

Similarly to the last Proposition, we have

\begin{proposition}
\label{prop:jpv tr}
If $A\in\TCi$ then $\jpvTr(A)$ is well-defined and $$\jpvTr(A)=\Tr(A).$$
\end{proposition}
\begin{proof}
Without loss of generality, set $j=1$ (the other case is obtained by exchanging the roles of the indices). As in the proof of the last Proposition, it is sufficient to prove the claim for $A$ self-adjoint. Thus, let $A$ be a self-adjoint operator. Notice that the operator $\chi_{1,L} A \chi_{1,L}$ is trace class because $\chi_{1,L}$ is a bounded operator and $A$ is trace class by hypothesis. 
Thus, both the trace of $\chi_{1,L} A \chi_{1,L}$ and of $A$ can be computed through the diagonal kernel and so one obtains
\begin{align*}
\Tr(A)&=\sum_{\V{n}\in\Z^2}\tr(A_{\V{n},\V{n}})=\sum_{n_2\in\Z}\sum_{n_1\in\Z}\tr(A_{\V{n},\V{n}})=\sum_{n_2\in\Z}\lim_{\substack{L\to\infty\\L\in 2\N+1}}\sum_{\substack{n_1\in\Z\\ \abs{n_1}\leq L/2}}\tr(A_{\V{n},\V{n}})\cr
&=\lim_{\substack{L\to\infty\\L\in 2\N+1}}\sum_{n_2\in\Z}\sum_{\substack{n_1\in\Z\\ \abs{n_1}\leq L/2}}\tr(A_{\V{n},\V{n}})=\lim_{\substack{L\to\infty\\L\in 2\N+1}}\Tr(\chi_{1,L}A\chi_{1,L})=1\text{-}\pvTr(A).
\end{align*}
In the last chain of equalities we have used in the order:
\begin{enumerate}[label=(\roman*), ref=(\roman*)]
\item the function $\Z^2\ni\V{n}\mapsto\tr(A_{\V{n},\V{n}})$ is in $\ell^1(\Z^2)$ by Lemma~\ref{lem:A sa tc Ann is l1}, therefore by Fubini's Theorem the series over $\Z^2$ does not depend on the order of summation over $n_1,n_2\in\Z$;
\item for every $n_2\in\Z$ the function $\Z\ni n_1\mapsto\tr(A_{\V{n},\V{n}})$ is in $\ell^1(\Z)$ by Lemma~\ref{lem:A sa tc Ann is l1} and Chebyshev's inequality, thus the series over $\Z$ can be computed through a particular exhaustion of $\Z$;
\item in view of the general Lebesgue's dominated convergence Theorem 
%(Section 4.4 Theorem 19 in \cite{RoydenFitzpatrick}) 
the limit over $L\to \infty$ and the series in $n_2\in\Z$ can be exchanged. 
\end{enumerate}
\end{proof}

In the following, we give two sufficient conditions for the existence of the trace per unit volume of an operator: 
the first one (periodicity) is well-known \cite{BoucletGerminet2005}, while the second one is, to our knowledge, new.  

%We will use the fact that, in view of the decomposition~\eqref{eqn:squareL is L^2 unit squares}, for $L\in 2\N+1$ one has         
%        \begin{equation}  \label{rem:decomposition of Lsquare}
%        \chi_L(\V{m})=\sum_{\substack{\V{n}\in\Z^2\\ \norm{\V{n}}_\infty\leq L/2}}\chi_1(\V{m}-\V{n})\quad\text{for all $\V{m}\in\Z^2$.}
%        \end{equation}
%

\smallskip

\begin{proposition}[Existence of TUV, condition I]
\label{prop:tau for per op}
Let $A$ be a {periodic} operator acting in $\Hi$. 
Then $\tau(A)$ is well-defined and 
$$
\tau(A)=\Tr(\chi_1 A \chi_1).
$$
\end{proposition}
\begin{proof}
The operator $\chi_L A \chi_L$ is trace class for every $L\in 2\N +1$, and its trace can be computed through the diagonal kernel. In view of periodicity,  one has $A_{\V{n},\V{n}}=A_{\V{0},\V{0}}$ for all $\V{n} \in \Z^2$. 
Therefore, by using the decomposition~\eqref{eqn:squareL is L^2 unit squares},  one obtains
\[
\Tr(\chi_L A \chi_L)=\sum_{\substack{\V{n}\in\Z^2\\ \norm{\V{n}}_\infty\leq L/2 }}\tr(A_{\V{n},\V{n}})=L^2\tr(A_{\V{0},\V{0}}).
\]
Hence  $\lim_{L\to\infty} \frac{1}{L^2} \Tr(\chi_L A \chi_L) = \tr(A_{\V{0},\V{0}}) = \Tr(\chi_1 A \chi_1)$, which concludes the proof.
\end{proof}

\begin{proposition}[Existence of TUV, condition II]
\label{prop:tau for notper: per+oddper}
Let $A,B$ be operators acting in $\Hi$  satisfying the following equation
\begin{equation}
\label{eqn:not per odd op}
A_{\V{m},\V{n}}=A_{\V{m}-\V{p},\V{n}-\V{p}}+g(\V{p}) \, B_{\V{m}-\V{p},\V{n}-\V{p}}\quad\text{for all $\V{m},\V{n},\V{p}\in\Z^2$},
\end{equation}
where $g\colon\Z^2\to \R$ is an {odd function in at least one variable}%%%%
%%%%%%%%%%
\footnote{
\label{foot:involution}
Namely, setting $(R_1 g)(n_1,n_2):= g(-n_1,n_2)$ and $(R_2 g)(n_1,n_2):= g(n_1,-n_2)$ for all $\V{n}\in\Z^2$, 
{one says that $g\colon\Z^2\to \R$ is an odd function in at least one variable if and only if there exists an index $j\in\set{1,2}$ such that $g(\V{n})=-(R_j g)(\V{n})$ for all $\V{n}\in\Z^2$.}
}. %%%% 
%%If $\chi_L A \chi_L$, $\chi_L B \chi_L$ are trace class for every $L\in 2\N+1$, 
Then $\tau(A)$ is well-defined and 
$$
\tau(A)=\Tr(\chi_1 A \chi_1).
$$
\end{proposition}
\begin{proof}
The operator $\chi_L A \chi_L$ is trace class, and we compute its trace through the diagonal kernel. In view of the equation~\eqref{eqn:not per odd op}, one has $A_{\V{n},\V{n}}=A_{\V{0},\V{0}}+ g(\V{n}) \, B_{\V{0},\V{0}}$. Therefore, using the decomposition~\eqref{eqn:squareL is L^2 unit squares}, we obtain
\begin{equation}
\label{eqn:tau for not per odd op}
\Tr(\chi_L A \chi_L)=\sum_{\substack{\V{n}\in\Z^2\\ \norm{\V{n}}_\infty\leq L/2 }}\tr(A_{\V{n},\V{n}})=L^2\tr(A_{\V{0},\V{0}})+\tr(B_{\V{0},\V{0}})\sum_{\substack{\V{n}\in\Z^2\\ \norm{\V{n}}_\infty\leq L/2 }}g     (\V{n}).
\end{equation}
Since the function $g$ is odd in at least one variable, there exists an index $j\in\set{1,2}$ such that $g(\V{n})=-(R_j g)(\V{n})$,
where $R_j$ is the corresponding reflection $^{\ref{foot:involution}}$. Denoting by $k$ the index different from $j$, we have
\[
\sum_{\substack{\V{n}\in\Z^2\\ \norm{\V{n}}_\infty\leq L/2 }}g(\V{n})=\sum_{\substack{n_k\in\Z\\ \abs{n_k}\leq L/2 }}\sum_{\substack{n_j\in\Z\\ \abs{n_j}\leq L/2}}g(\V{n})= 0. %\sum_{\substack{n_k\in\Z\\ \abs{n_k}\leq L/2 }}0=0.
\]
As ${\tr(B_{\V{0},\V{0}})}={\Tr({\chi_1 B \chi_1})}$ is finite (since $\Ran \chi_1$ is finite-dimensional),  
the second summand on the right-hand side of~\eqref{eqn:tau for not per odd op} vanishes. This concludes the proof.
\end{proof}

\goodbreak

\begin{proposition}[Conditional cyclicity of the trace per unit volume]
\label{prop:cycl of tau}
Let $A,B$ be periodic operators acting in $\Hi$.  
%such that $\chi_L AB\chi_L$ and $\chi_L B A\chi_L$ are trace class for every $L\in 2\N+1$, 
Then $$\tau(AB)=\tau(BA).$$
\end{proposition}

\begin{proof}
Applying Proposition~\ref{prop:tau for per op} and computing the trace of $\chi_1 AB \chi_1$ through the diagonal kernel, we have
\begin{equation}
\label{eqn:cycl of tau step1}
\tau(AB)=\Tr(\chi_1 AB \chi_1)=\sum_{\V{n}\in\Z^2}\tr(A_{\V{0},\V{n}}B_{\V{n},\V{0}}).
\end{equation}
We rewrite the term on right-hand side of the last equation. Using the periodicity of the operators $A$ and $B$, and the invariance of $\Z^2$ under the reflection $\V{n}\mapsto -\V{n}$, we obtain
\begin{align}
\label{eqn:cycl of tau step2}
\sum_{\V{n}\in\Z^2}\tr(A_{\V{0},\V{n}}B_{\V{n},\V{0}})&=\sum_{\V{n}\in\Z^2}\tr(A_{-\V{n},\V{0}}B_{\V{0},-\V{n}})\cr
&=\sum_{\V{n}\in\Z^2}\tr(A_{\V{n},\V{0}}B_{\V{0},\V{n}}).
\end{align}
As $\tr(\,\cdot\,)$ acts on a finite-dimensional Hilbert space, one has $\tr(A_{\V{n},\V{0}}B_{\V{0},\V{n}})=\tr(B_{\V{0},\V{n}}A_{\V{n},\V{0}})$ for every $\V{n}\in\Z^2$. Therefore, in view of Proposition~\ref{prop:tau for per op}, we have
\begin{equation}
\label{eqn:cycl of tau step3}
\sum_{\V{n}\in\Z^2}\tr(A_{\V{n},\V{0}}B_{\V{0},\V{n}})=\sum_{\V{n}\in\Z^2}\tr(B_{\V{0},\V{n}}A_{\V{n},\V{0}})=\Tr(\chi_1 B A\chi_1)=\tau(BA).
\end{equation}
Plugging equation \eqref{eqn:cycl of tau step3} into \eqref{eqn:cycl of tau step1}, 
the proof is concluded.
\end{proof}

\section{Localization properties of near-sighted operators }
\label{Sec:Trace near-sighted}

In this Section, we consider the peculiar localization properties of operators which are near-sighted, see Definition \ref{defn:op near-sighted}, and their relation with the trace class condition. 

\begin{remark}
\label{rem:boundness techn}
For the later purposes, it is convenient to recall some elementary but useful tools to establish the boundedness of an operator $A$ acting on $\Hi$:
\begin{enumerate}[label=(\roman*), ref=(\roman*)]
\item \label{item:Hol est} the H\"olmgren's estimate 
\[
\norm{A}\leq \max \( \sup_{\V{m} \in \Z^2} \sum_{\V{n}\in \Z^2} \abs{A_{\V{m},\V{n}}} \, , \, \sup_{\V{n} \in \Z^2} \sum_{\V{m}\in \Z^2} \abs{A_{\V{m},\V{n}}} \);
\]
\item \label{item:geom series} the convergence of the following series: for every $\lambda>0$ and $ j\in\set{1,2}$, we have
\[
\sum_{n_j \in \Z} \E^{-\lambda\abs{n_j}}=  \dfrac{1+ \ee^{-\lambda}}{1-\E^{-\lambda} } < \infty.
\]
\end{enumerate}
\end{remark}	

%\begin{remark}
%\label{rem:nearsighted implies bounded}
%Notice that if $A$ is near-sighted with range $\zeta$ then $A$ is bounded. %%%
%\footnote{\NB{From our perspective this is a tautology, since we used the \underline{boundedness} of $A$ to define its matrix kernel, which then is used in the Definition \ref{defn:op near-sighted}. I would skip this Remark. (A different situation appear when $A$ is defined by its integral kernel) \textcolor{red}{C: I agree, we can skip it}}
%}\  %%%%
%Indeed, in view of Remark~\ref{rem:boundness techn}, we have
%\[
%\norm{A}\leq C\sup_{\V{m}\in\Z^2}\sum_{\V{n} \in \Z^2} \E^{-\frac{1}{\zeta}\norm{\V{m}-\V{n}}_1} =C\sum_{\V{n} \in \Z^2} \E^{-\frac{1}{\zeta}\norm{\V{n}}_1}=C \( \dfrac{1+ \ee^{-\frac{1}{\zeta}}}{1-\E^{-\frac{1}{\zeta}} }\)^2 < \infty.
%\]
%\end{remark}

Preliminary, we recall some results which are useful to establish the trace class property in the discrete case  \cite{ElgartGrafSchenker}.

\begin{definition}
\label{defn:1lip funct}
A function $f : \Z^2 \rightarrow \mathbb R$ is called \textbf{1-Lipschitz} if and only if it satisfies
\[
\abs{f(\V{m})-f(\V{n})}\leq \norm{\V{m}-\V{n}}_1\text{ for all $\V{m},\V{n}\in\Z^2$}.
\]
\end{definition}

\begin{lemma}
\label{lem:exp 1lip A exp -1lip}
If $A$ is a near-sighted operator acting in $\Hi$ with range $\zeta$, then
$$
\E^{\pm\lambda f} A \E^{\mp\lambda f}\quad\text{is bounded with } \norm{\E^{\pm\lambda f} A \E^{\mp\lambda f}}\leq \( \dfrac{1+ \ee^{-\(\frac{1}{\zeta}-\lambda \)} }{1-\E^{-\(\frac{1}{\zeta}-\lambda\)} }\)^2,
$$
for every $0\leq \lambda < 1/\zeta$ and every 1-Lipshitz function $f$. 
\end{lemma}
\begin{proof}
%By triangular inequality, we have
%\[
%\norm{ \ee^{\pm\lambda f} A \ee^{\mp\lambda f} } \leq \norm{\ee^{\pm\lambda f} A \ee^{\mp\lambda f}- A} + \norm{A},
%\]
%being $A$ bounded by Remark~{rem:nearsighted implies bounded}, it is enough to prove that $\norm{\ee^{\pm\lambda f} A \ee^{\mp\lambda f}- A} <\infty$. To obtain last bound, we proceed as follows.

For every $\V{m}\in\Z^2$ we compute 
\begin{align*}
\sum_{\V{n}\in \Z^2} \abs{\(\E^{\pm \lambda f} A \E^{\mp\lambda f}\)_{\V{m},\V{n}}}&= \sum_{\V{n}\in \Z^2} \abs{\E^{\pm\lambda f(\V{m})}A_{\V{m},\V{n}}\E^{\mp\lambda f(\V{n})}} \cr
&=\sum_{\V{n}\in \Z^2} \big|A_{\V{m},\V{n}}\E^{\pm\lambda \( f(\V{m}) - f(\V{n}) \)}\big|\cr
&\leq C\sum_{\V{n}\in \Z^2}  \ee^{- \frac{1}{\zeta} \norm{\V{m-n} }_1} \ee^{\lambda \abs{f(\V{m})-f(\V{n}) }} \cr
&\leq C  \sum_{\V{n}\in \Z^2}   \ee^{- \frac{1}{\zeta} \norm{\V{m-n} }_1} \ee^{\lambda \norm{\V{m-n} }_1}\cr
&=C  \sum_{\V{n}\in \Z^2}   \ee^{-\(\frac{1}{\zeta}-\lambda\) \norm{\V{n}}_1},
\end{align*}
where we have used the near-sightedness of $A$, the inequality $\abs{\E^{a}}\leq \E^{\abs{a}}$ for all $a\in\R$, the fact that $f$ is a 1-Lipschitz function and $\Z^2$ is invariant under $\Z^2$-translation. On the right-hand side of the last inequality, the series is finite as long as $\lambda < \frac{1}{\zeta}$. 
After an analogous computation which considers the sum over $\V m \in \Z^2$,
in view of Remark~\ref{rem:boundness techn}, we obtain that for every $0\leq \lambda < \frac{1}{\zeta}$
\[
\norm{\E^{\pm \lambda f} A \, \E^{\mp\lambda f}}\leq \( \dfrac{1+ \ee^{-\(\frac{1}{\zeta}-\lambda \)} }{1-\E^{-\(\frac{1}{\zeta}-\lambda\)} }\)^2.
\]
\end{proof}

\goodbreak

\begin{definition}
Let $A\in\B(\Hi)$. For $j\in\set{1,2}$ and $\alpha>0$ we say that $A$ is \textbf{$\alpha$-confined in $j^{\mathrm{th}}$-direction}\footnote{In the terminology of 
\cite{ElgartGrafSchenker}.} if and only if
$$
A \, \E^{\alpha |X_j|}\text{ is bounded}.
$$
\end{definition}

\label{rem:confined op}
\noindent Clearly, if $A$ is $\alpha$-confined in $j^{\mathrm{th}}$-direction for some $\alpha>0$ and $j\in\set{1,2}$, then $A$ is $\lambda$-confined in $j^{\mathrm{th}}$-direction for every $0<\lambda\leq \alpha$. 
%It is straightforward noticing that 
%$$
%\norm{A \E^{\lambda |X_j|}}\leq \norm{A \E^{\alpha |X_j|}}\norm{ \E^{\(\lambda-\alpha \)|X_j|}}=\norm{A \E^{\alpha |X_j|}}<\infty.
%$$

\begin{remark}
\label{rem:alg identity}
Here, we notice a simple algebraic identity which will be useful to recall in different proofs.
For any $A,B\in\BH$ we have
\[
[B,A]=BA-BAB-AB+BAB=BA(\Id-B)-(\Id-B)AB.
\]
\end{remark}

\begin{lemma}
\label{lem:[nearsighted,switch] is conf}
Let $A$ be a near-sighted operator acting in $\Hi$, with range $\zeta$,  and let $\Lambda_j$ be a switch function in $j^{\mathrm{th}}$-direction. Then
$$
[\Lambda_j, A]\text{ is $\alpha$-confined in $j^{\mathrm{th}}$-direction,}
$$
for some $0<\alpha<1/\zeta$.
\end{lemma}
\begin{proof}
Using Remark~\ref{rem:alg identity}, we have for $0<\alpha\leq \frac{1}{\zeta}$
\begin{align}
\label{eqn:[nearsighted,switch] is conf step1}
[\Lambda_j, A] \, \ee^{\alpha|X_j|} = \Lambda_j A (1-\Lambda_j) \, \ee^{\alpha|X_j|} - (1-\Lambda_j) A \Lambda_j \, \ee^{\alpha|X_j|}.
\end{align}
We analyse the first summand on the right-hand side of the last equation. We have
\begin{equation}
\label{eqn:[nearsighted,switch] is conf step2}
\Lambda_j A (1-\Lambda_j) \ee^{\alpha|X_j|} = \Lambda_j A (1-\Lambda_j) \chi_{\set{n_j < n_-}} \ee^{\alpha |X_j|}	+ \Lambda_j A (1-\Lambda_j) \chi_{\set{n_- \leq n_j < n_+}}\ee^{\alpha|X_j|}.
\end{equation}
The second summand on the right-hand side of the last equation is bounded since $\chi_{\set{n_- \leq n_j < n_+}}$ is compactly supported in direction $j$ and $A$ is bounded. %by Remark~\ref{rem:nearsighted implies bounded}. 
On the other hand, for the first summand on the right-hand side of \eqref{eqn:[nearsighted,switch] is conf step2}, we have to split the case either $n_-\leq 0$ or $n_->0$. 

For $n_-\leq 0$, we obtain
\[
\Lambda_j A (1-\Lambda_j) \chi_{\set{n_j < n_-}} \ee^{-\alpha X_j}=\Lambda_j\ee^{-\alpha X_j} \cdot \ee^{\alpha X_j}A \ee^{-\alpha X_j} \cdot (1-\Lambda_j) \chi_{\set{n_j < n_-}},
\]
which is bounded because $\norm{\Lambda_j\ee^{-\alpha X_j}}=\norm{\Lambda_j \chi_{\set{n_j \geq n_-}}\ee^{-\alpha X_j}}\leq \E^{-\alpha n_-}$ and $\ee^{\alpha X_j}A \ee^{-\alpha X_j}$ is bounded by Lemma~\ref{lem:exp 1lip A exp -1lip}.

For $n_->0$, we obtain
\[
\Lambda_j A (1-\Lambda_j) \chi_{\set{n_j < n_-}} \ee^{\alpha |X_j|}=\Lambda_j A (1-\Lambda_j) \chi_{\set{n_j \leq 0}} \ee^{-\alpha X_j}+\Lambda_j A (1-\Lambda_j) \chi_{\set{0<n_j < n_-}} \ee^{\alpha |X_j|},
\]
which is bounded because on the right-hand side the first-summand is bounded by analogy to the previous case, and the second summand is bounded because $\chi_{\set{0<n_j < n_-}}$ is compactly supported in direction $j$ and $A$ is bounded. %by Remark~\ref{rem:nearsighted implies bounded}.

Therefore, $\Lambda_j A (1-\Lambda_j) \ee^{\alpha|X_j|}$ is bounded. Proceeding similarly for the second term on the right-hand side of the equality~\eqref{eqn:[nearsighted,switch] is conf step1}, we deduce that $[\Lambda_j, A] \ee^{\alpha|X_j|}$ is bounded.
\end{proof}

\begin{proposition}
\label{prop:suff traceclass cond ABC}
Let $j\neq k\in\set{1,2}$. If $A$ is $\alpha$-confined in the $j^{\mathrm{th}}$-direction, $B$ is a bounded operator such that $B^*$ is $\beta$-confined in the $k^{\mathrm{th}}$-direction and $C$ is an operator such that
$$
\E^{-\alpha |X_1|}C\E^{\alpha |X_1|}\text{ is bounded or }\E^{\beta |X_2|}C\E^{-\beta |X_2|}\text{is bounded},
$$ then $ACB$ is trace class.
\end{proposition}
\begin{proof}
Without loss of generality, we suppose that $j=1$ and $k=2$.  %the other case is obtained exchanging the roles of the indices.

Assume that $\E^{-\alpha |X_1|}C\E^{\alpha |X_1|}$ is bounded. We have
\[
ACB=A\E^{\alpha |X_1|}\cdot\E^{-\alpha |X_1|}C\E^{\alpha |X_1|}\cdot \E^{-\alpha |X_1|}\E^{-\beta |X_2|}\cdot \E^{\beta |X_2|}B,
\]
which is trace class. Indeed, on the right-hand side of the last equality the first factor is bounded by hypothesis and the second one is bounded by assumption. 

For the fourth factor, in view of $T^*S^*\subseteq {\(ST\)}^*$ for any $S$ and $T$ closed, densely defined operators in $\Hi$, we have $\norm{\E^{\beta |X_2|}B}\leq \norm{{\(B^*\E^{\beta |X_2|}\)}^*}=\norm{B^*\E^{\beta |X_2|}}<\infty$ by hypothesis. 

The third factor is trace class. Indeed, 
\[
\Tr(\E^{-\alpha |X_1|}\E^{-\beta |X_2|})=\sum_{\V{n}\in\Z^2}\E^{-\alpha |n_1|}\E^{-\beta |n_2|}= \( \dfrac{1+ \ee^{-\alpha}}{1-\E^{-\alpha} }\) \( \dfrac{1+ \ee^{-\beta}}{1-\E^{-\beta} }\)<\infty.
\]

On the other hand, assume that $\E^{\beta |X_2|}C\E^{-\beta |X_2|}$ is bounded. Writing
\[
ACB=A\E^{\alpha |X_1|}\cdot\E^{-\alpha |X_1|}\E^{-\beta |X_2|}\cdot\E^{\beta |X_2|}C \E^{-\beta |X_2|}\cdot \E^{\beta |X_2|}B,
\]
we can reason similarly to the previous case. This concludes the proof.
\end{proof}

\begin{remark}[Discrete vs continuum models]
\label{rem:on trace class prop in disc vs cont}
This strategy to establish trace class property is based on the fact that $\E^{-\alpha |X_1|}\E^{-\beta |X_2|}$ is trace class for some $\alpha,\beta>0$, a property which holds true for the discrete models considered in this paper, 
but not for continuum models. 
In other words, this property is rooted in the underlying ultraviolet cutoff of the discrete models. The generalization to  continuum models would require further assumptions on the operators such as localization in energy. 
%Indeed, given a bounded subset $\Omega \subset \R^2$ the operators $\chi_{\Omega}$ or $\chi_{\Omega}P^{\perp}$ are not trace class in $L^2(\R^2)$.
\end{remark}

\begin{remark}
One might naively think that $[P, \Lambda_1 S_z]$ is $\alpha$-confined in the $1^{\mathrm{st}}$-direction for some $\alpha>0$, since $P$ is near-sighted and $S_z$ acts non-trivially only on the $\C^2$ sector. This is not true in general. Indeed, we have
\begin{align*}
[P, \Lambda_1 S_z]=[P,S_z]\Lambda_1+S_z[P,\Lambda_1].
\end{align*}
On the right-hand side, the second summand is confined by Lemma~\ref{lem:[nearsighted,switch] is conf}, while the first summand has no reason to be confined, since $[P,S_z]$ is a priori only a bounded operator which does not have \emph{decreasing properties in space}. Consequently $\G_K^{s_z}$ is not trace class in general, since it is not confined in the  $1^{\mathrm{st}}$-direction. This is why we had to introduce the directional principal value trace in the definition of $G_K^{s_z}$.
\end{remark}

\goodbreak

\section{Proof of the main results}
\label{sect:main results}

Recall that the Hamiltonian operator $H$ satisfies Assumption~\ref{ass:H}. Namely,  $H$ is near-sighted, periodic and with a spectral gap $\Delta$. For $\mu \in \Delta$,  $P=\chi_{(-\infty,\mu)}(H)$ is the corresponding Fermi projection. Under these hypotheses, it is well-known that
\begin{lemma}[\cite{AizenmanWarzel,AizenmanGraf,Kirsch}]
\label{lem:P is nearsighted}
The Fermi projection $P$ is near-sighted. 
\end{lemma}

\noindent We denote the range of $P$ by $\zeta_P$. 
%This fact essentially follows from the discrete version of Combes-Thomas theory, see for example.	
Note also that $P^\perp=\Id-P$ is near-sighted. %Indeed, we have
%\begin{equation}
%\label{eqn:Pperp is nearsighted}
%|P^\perp_{\V{m},\V{n}}| = |(\Id-P)_{\V{m},\V{n}}| \leq \delta_{\V{m},\V{n}} + C \ee^{-\frac{1}{\zeta_P}\norm{\V{m}-\V{n}}_1} \leq (1+C)\ee^{-\frac{1}{\zeta_P}\norm{\V{m}-\V{n}}_1}.
%\end{equation}

\begin{proposition}
\label{prop:[nearsighted,Xj] is bounded}
If $A$ is a near-sighted operator acting in $\Hi$, then we have that
$$
[A,X_j]\text{ is bounded for $j\in\set{1,2}$.}
$$
\end{proposition}
\begin{proof}
Fix $j=1$ (the other case is obtained by replacing the index $1$ with $2$).
For every $\V{m}\in\Z^2$ we compute
\begin{align}
\label{eqn:[A,Xj] is bounded step1}
\sum_{\V{n}\in \Z^2} \abs{\( [A,X_1]\)_{\V{m},\V{n}}}&= \sum_{\V{n}\in \Z^2} \abs{ m_1-n_1 }\abs{A_{\V{m},\V{n}}}
\leq C\sum_{\V{n}\in \Z^2} \E^{ -\frac{1}{\zeta} \norm{\V{m}-\V{n}}_1  }\abs{m_1-n_1}\cr
&=C\sum_{\V{n}\in \Z^2}\E^{ -\frac{1}{\zeta} \abs{m_2-n_2}  } \E^{ -\frac{1}{\zeta} \abs{m_1-n_1}  }\abs{m_1-n_1}\cr
&=C\dfrac{1+ \ee^{-\zeta}}{1-\E^{-\zeta} }\sum_{n_1\in \Z} \E^{ -\frac{1}{\zeta} \abs{n_1}  }\abs{n_1},
\end{align}
where we used in the last step Remark~\ref{rem:boundness techn}~\ref{item:geom series} and the invariance of $\Z$ under $\Z$-translations. 
Clearly, the series on the right-hand side of \eqref{eqn:[A,Xj] is bounded step1} is convergent. 
A similar computation involving the sum over $\V m \in \Z^2$ and Remark~\ref{rem:boundness techn}~\ref{item:Hol est} imply the thesis.
\end{proof}

\begin{lemma}
\label{lem:algeb prop [A,X_j] and [AS,X_j]}
If $A$ is a periodic operator acting in $\Hi$ and $S$ is an operator acting non-trivially on $\C^N\otimes\C^2$ only, then for $j\in\set{1,2}$ we have the following
\begin{enumerate}[label=(\roman*), ref=(\roman*)]
\item \label{item:[A,X_j]} the operator $[A,X_j]$ is periodic, namely 
$$
{\( [A,X_j]\)}_{\V{m},\V{n}}={\( [A,X_j]\)}_{\V{m}-\V{p}, \V{n}-\V{p}}\quad\text{for all $\V{m},\V{n},\V{p}\in\Z^2$};
$$ 

\item \label{item:[A,X_jS]} the operator $[A,X_jS]$ satisfies 
$$
{\([A,X_jS]\)}_{\V{m},\V{n}}={\([A,X_jS]\)}_{\V{m-p},\V{n-p}}-p_j{\( [A,S]  \)}_{\V{m-p},\V{n-p}}\quad\text{for all $\V{m},\V{n},\V{p}\in\Z^2$}.
$$
\end{enumerate}
\end{lemma}
\begin{proof}
Recall that we denote by $T_{\V{p}}$ the translation operator by the vector $\V{p}\in\Z^2$, acting on $\ell^2(\Z^2)$, 
or similarly on $\Hi$,  as
$$
{\(T_{\V{p}}\varphi\)}_{\V{n}}:=\varphi_{\V{n-p}}\quad\text{for all $\varphi\in\ell^2(\Z^2)$}.
$$

{\it \ref{item:[A,X_j]}} By Jacobi identity, we have
\begin{equation}
\label{eqn:[[A,X_j],T_p]=0}
[[A,X_j],T_{\V{p}}]=-[[T_{\V{p}},A],X_j]-[[X_j,T_{\V{p}}],A]=-[[X_j,T_{\V{p}}],A]=-p_j[T_{\V{p}},A]=0,
\end{equation}
where we have used the periodicity of $A$ and the identity
\begin{equation}
\label{eqn:[X_j,Tp]}
[X_j,T_{\V{p}}]=p_jT_{\V{p}}\quad\text{for all $\V{p}\in\Z^2$}.
\end{equation}
%{We set $\delta_{\V{n}}^{(k)}  := \delta_{\V{n}} \otimes e_{k} \in \Hi$, where $\set{e_k}_{k\in\set{1,\dots,2N}}$ is an orthonormal basis of $\C^{2N} \simeq \C^N \otimes \C^2$. \NB{C: I updated the expression below, without red.}} 
Recalling the definition \eqref{Def: delta basis}, by the commutation relation~\eqref{eqn:[[A,X_j],T_p]=0} for every $\V{m},\V{n},\V{p}\in\Z^2$ we obtain
\begin{align*}
{\( [A,X_j]\)}_{\V{m},\V{n}}^{(i),(j)}&=
\inner{\delta_{\V{m}}^{(i)} }{ [A,X_j] \, \delta_{\V{n}}^{(j)}} 
=\inner{\delta_{\V{m}}^{(i)} }{ T_{\V{p}}^* [A,X_j] T_{\V{p}} \delta_{\V{n}}^{(j)} }\cr
&=\inner{ \delta_{\V{m-p}}^{(i)} }{ [A,X_j]   \delta_{\V{n-p}}^{(j)} }={\( [A,X_j]\)}_{\V{m-p},\V{n-p}}^{(i),(j)}.
\end{align*}

{\it \ref{item:[A,X_jS]}} By Leibniz rule, we have
\begin{equation}
\label{eqn:Leibniz [A,X_jS]}
[A,X_jS]=[A,X_j]S+X_j[A,S].
\end{equation}
On the right-hand side of the last equation the first summand is periodic, as it is the product of an operator which is periodic by the previous claim~\ref{item:[A,X_j]} and $S$, which acts non-trivially only in the sector $\C^N\otimes\C^2$. Instead, the second summand is such that, in view of the identity~\eqref{eqn:[X_j,Tp]} and the periodicity of $[A,S]$,  
\[
[X_j[A,S],T_{\V{p}}]=[X_j,T_{\V{p}}][A,S]=p_jT_{\V{p}}[A,S].
\]
Therefore, using the decomposition~\eqref{eqn:Leibniz [A,X_jS]}, the claim~\ref{item:[A,X_j]}, the previous relation and the periodicity of $[A,S]$, for every $\V{m},\V{n},\V{p}\in\Z^2$ we have
\begin{align*}
{\( [A,X_jS]\)}_{\V{m},\V{n}}^{(i),(j)}
&={\( [A,X_j]S\)}_{\V{m-p},\V{n-p}}^{(i),(j)}+
\inner{ \delta_{\V{m}}^{(i)} }{ T_{\V{p}}^* X_j[A,S] T_{\V{p}} \, \delta_{\V{n}}^{(j)} } 
- p_j  \inner{\delta_{\V{m}}^{(i)}}{ [A,S] \, \delta_{\V{n}}^{(j)} }\cr
&={\( [A,X_j]S\)}_{\V{m-p},\V{n-p}}^{(i),(j)}+\inner{ \delta_{\V{m-p}}^{(i)} }{  X_j[A,S] \, \delta_{\V{n-p}}^{(j)} }
-p_j  \inner{ \delta_{\V{m-p}}^{(i)} }{ [A,S] \, \delta_{\V{n-p}}^{(j)} }\cr
&={\( [A,X_j S]\)}_{\V{m-p},\V{n-p}}^{(i),(j)}-p_j{\( [A,S]\)}_{\V{m-p},\V{n-p}}^{(i),(j)}.
\end{align*}
\end{proof}

\subsection{Proof of Lemma~\ref{lem:spin torque}}
The operator $\mathcal{T}_{s_z}$ is periodic, since $[P,X_2]$ is so by Lemma~\ref{lem:algeb prop [A,X_j] and
 [AS,X_j]}~\ref{item:[A,X_j]} and the other operators involved in its definition are periodic. It is also bounded as $[P,X_2]$ is so by Proposition~\ref{prop:[nearsighted,Xj] is bounded} and Lemma~\ref{lem:P is nearsighted}, and the other operators are bounded.

As $\mathcal{T}_{s_z}$ is periodic and bounded, 
%$\chi_L \mathcal{T}_{s_z}\chi_L$ is trace class for every $L\in 2\N+1$, 
one concludes the proof by invoking Proposition~\ref{prop:tau for per op}.\qed

\subsection{Proof of Theorem~\ref{thm:main1}}
In view of Lemma~\ref{lem:spin torque}, one has that $\tau(\mathcal{T}_{s_z})$ is well-defined. By algebraic manipulations and Proposition~\ref{prop:cycl of tau}, one obtains
\begin{align*}
\tau(\mathcal{T}_{s_z})&=\iu\tau(P S_z P^\perp [P,X_2])+ \iu \tau(P[P,X_2]P^\perp S_z P)\cr
&=\iu\tau( S_z P^\perp [P,X_2]P+S_z P[P,X_2]P^\perp)=\iu\tau( S_z  [P,X_2])=\iu\tau([S_z P,X_2]).
\end{align*}
As mentioned above, $S_z  [P,X_2] = [S_z P, X_2]$ is a periodic bounded operator. 
Hence, in view of Propositions~\ref{prop:tau for per op}~and~\ref{prop:cycl of trace}, the commutation relation $[ X_2,\chi_1]=0$ and the identity $\chi_1^2=\chi_1$, we rewrite the term on the right-hand side of the last equality as
\begin{align*}
\iu\tau([S_z P,X_2]) & = 
\iu \Tr(\chi_1  S_z P  \chi_1 X_2   \chi_1)-\iu \Tr(\chi_1   X_2\chi_1 S_z P  \chi_1)= \cr
& = \iu \Tr(\chi_1  S_z P  \chi_1 X_2   \chi_1) -\iu \Tr(\chi_1 S_z P  \chi_1   X_2\chi_1)=0.
\end{align*}
\qed

\subsection{Proof of Theorem~\ref{thm:main2}}
\begin{itemize}
\item[Part \ref{item:defn GKs}:] 
Assume that $G_K^s(\Lambda_1,\Lambda_2)$ (exists and) is finite for a particular switch function $\Lambda_1$. 
%Now, we focus attention on the map $G_K^s(\,\cdot\,,\Lambda_2)$ for every fixed switch function $\Lambda_2$ in the $2^{\mathrm{nd}}$-direction. 
Given another switch functions $\Lambda_1^{\prime}$, we set $\Delta \Lambda_1 = \Lambda_1 - \Lambda_1^{\prime}$. 
By algebraic manipulations, using $P^2=P$ and $P^\perp = \Id-P$, we have
\begin{align}
\label{eqn:GKs indep of Lambda1 step1}
G_K^s(\Delta \Lambda_1, \Lambda_2) &= 1\text{-}\pv \Tr\( \ii P [[P, \Delta\Lambda_1 S_z], \, [P,\Lambda_2]]P\) \cr
& =  1\text{-}\pv \Tr \(\ii [P, \Delta\Lambda_1 S_z] P^\perp [P,\Lambda_2] - \ii [P, \Lambda_2 ] P^\perp [P,\Delta\Lambda_1 S_z]\)\cr
&=1\text{-}\pv \Tr \(\ii P\Delta\Lambda_1 S_z P^\perp [P,\Lambda_2] + \adj \),
\end{align}
where $\pm\adj$ means that the adjoint of the sum of all operators to the left is added, respectively subtracted.
Notice that $\ii P\Delta\Lambda_1 S_z P^\perp [P,\Lambda_2]=\ii P\Delta\Lambda_1 S_z [P,\Lambda_2] P$ is trace class. Indeed, Proposition~\ref{prop:suff traceclass cond ABC} applies to $A=\Delta\Lambda_1$, which is $\alpha$-confined in the $1^{\mathrm{st}}$-direction for some $\alpha>0$, $B^*=[\Lambda_2,P]$, which is $\beta$-confined in the $2^{\mathrm{nd}}$-direction for some $0<\beta<1/\zeta_P$ by Lemma~\ref{lem:[nearsighted,switch] is conf} and Lemma~\ref{lem:P is nearsighted}, and $C=S_z$, and thus we deduce that $\Delta\Lambda_1 S_z [P,\Lambda_2]$ is trace class and therefore $\ii P\Delta\Lambda_1 S_z [P,\Lambda_2] P$ is so, as $P$ is bounded. As $\TC$ is closed under adjointness, we have that the argument of the directional principal value trace on the right-hand side of \eqref{eqn:GKs indep of Lambda1 step1} is trace class, with the two summands separately trace class. 

%\NB{Questo serve per separare i due addendi in (5.7)}

Therefore, in view of Proposition~\ref{prop:jpv tr} we obtain
\begin{equation}
\label{eqn:GKs indep of Lambda1 step2}
1\text{-}\pv \Tr \(\ii P\Delta\Lambda_1 S_z P^\perp [P,\Lambda_2] + \adj\)=\Tr\(\ii P\Delta\Lambda_1 S_z P^\perp [P,\Lambda_2] + \adj\).
\end{equation}
Using in the order the linearity and the cyclicity of the trace (apply Proposition~\ref{prop:cycl of trace} under the hypothesis $A\in\TC$ and $B\in\BH$), we obtain 
\begin{align}
\label{eqn:GKs indep of Lambda1 step3}
\Tr\(\ii P\Delta\Lambda_1 S_z P^\perp [P,\Lambda_2] + \adj\)&=\Tr\(-\ii P\Delta\Lambda_1 S_z P^\perp \Lambda_2 P + \adj\)\cr
&=-\iu\Tr\( P\Delta\Lambda_1 S_z P^\perp \Lambda_2 P \)+\iu\Tr\( P\Lambda_2 P^\perp \Delta\Lambda_1 S_z P\)\cr
&=-\iu\Tr\( \Delta\Lambda_1 S_z P^\perp \Lambda_2 P \)+\iu\Tr\( P\Lambda_2 P^\perp \Delta\Lambda_1 S_z \),\cr
\end{align}
as $\Delta\Lambda_1 S_z P^\perp \Lambda_2 P=-\Delta\Lambda_1 S_z[ P,\Lambda_2]P\in\TCi$ for the previous analysis and $P\in\BH$, and a similar reasoning shows that the operator {$P\Lambda_2 P^\perp \Delta\Lambda_1 S_z$} is also trace class.

In view of Proposition~\ref{prop:cycl of trace} (in the hypothesis $A,B\in\BH$ such that $AB$ and $BA$ are both in $\TC$), 
we rewrite the second summand on the right-hand side of the last equation as
\begin{equation}
\label{eqn:GKs indep of Lambda1 step4}
\iu\Tr\( P\Lambda_2 P^\perp \Delta\Lambda_1 S_z \)=\iu\Tr\( \Delta\Lambda_1 S_z P\Lambda_2 P^\perp\),
\end{equation}
since $\Delta\Lambda_1 S_z P\Lambda_2 P^\perp=\Delta\Lambda_1 S_z [P,\Lambda_2] P^\perp$ is in $\TCi$ by Proposition~\ref{prop:suff traceclass cond ABC} applied to $A=\Delta\Lambda_1$, 
%which is $\alpha$-confined in the $1^{\mathrm{st}}$-direction for some $\alpha>0$, 
$B^*=[\Lambda_2,P]$, 
%which is $\beta$-confined in the $2^{\mathrm{nd}}$-direction for some $0<\beta<1/\zeta_P$ 
and $C=S_z$.
Plugging the equations \eqref{eqn:GKs indep of Lambda1 step2}, \eqref{eqn:GKs indep of Lambda1 step3}, \eqref{eqn:GKs indep of Lambda1 step4} in \eqref{eqn:GKs indep of Lambda1 step1} and finally using Remark~\ref{rem:alg identity}, we have
\begin{align*}
G_K^{s_z}(\Delta \Lambda_1, \Lambda_2)&=\iu\Tr\big( \Delta\Lambda_1 S_z \(P\Lambda_2 P^\perp -P^\perp \Lambda_2 P\)  \big)\cr
&=\iu\Tr\big( \Delta\Lambda_1 S_z [P,\Lambda_2 ]  \big)\cr
&=\iu \sum_{\V{m} \in \mathbb  Z^2} \tr \big( \Delta\Lambda_1(m_1) S_z \Lambda_2(m_2) P_{\V{m},\V{m}} -  \Delta\Lambda_1(m_1) S_z  P_{\V{m},\V{m}}\Lambda_2(m_2) \big) \cr
& = 0.
\end{align*}
This shows that whenever $G_K^s(\Lambda_1,\Lambda_2)$ is finite, also 
$G_K^s(\Lambda_1^{\prime},\Lambda_2)$ is finite and equals the former one.
%\footnote{\gp{{\bf Remark for internal use:} notice that it might happen, in principle, that  $G_K^s(\Lambda_1,\Lambda_2)$ and $G_K^s(\widetilde \Lambda_1,\Lambda_2)$ do not exist, separately,  but   $G_K^s(\Lambda_1 - \widetilde \Lambda_1,\Lambda_2) =0$}.
%}\  %%%%
This concludes the proof of Part~\ref{item:defn GKs}.

\medskip

\bigskip

\item[Part \ref{item:defn sigmaKs}:] The equation~\eqref{eqn:Sigma notper: per+oddper} is implied by Lemma~\ref{lem:algeb prop [A,X_j] and [AS,X_j]}~\ref{item:[A,X_jS]}. Once established \eqref{eqn:Sigma notper: per+oddper}, Proposition~\ref{prop:tau for notper: per+oddper} concludes the proof of Part~\ref{item:defn sigmaKs}.

\bigskip

%%%% Redundant argument %%%%%%%%%%%%%%%%%%%%%%%%%%%%%%%
%In view of Leibniz rule, the operator $\Sigma_K^{s_z}$ can be rewritten as
%\begin{align*}
%\Sigma_K^{s_z}=\iu P [[P, X_1 S_z],  [P,X_2]] P&=\iu P [[P, X_1 ]S_z,  [P,X_2]] P+\iu [P, X_1][[P,  S_z],  [P,X_2]] P +\cr
%&\phantom{=}+\iu X_1P[[P,  S_z],  [P,X_2]] P+\iu P[X_1,P][X_2,P][P,S_z]P+\cr
%&\phantom{=}+\iu P[X_1,P]P[X_2,P]S_z P-\iu P[X_1,P]P S_z [X_2,P] P+\cr
%&\phantom{=}+\iu P[X_1,P]P[P,  S_z][X_2,P]+\iu P[X_1,P]P[P,S_z]PX_2+\cr
%&\phantom{=}-\iu [P,X_2][X_1,P]P[P,S_z]P-\iu X_2 P[X_1,P]P[P,S_z]P.
%\end{align*}
%As for $j\in\set{1,2}$, $\chi_L X_j=\chi_L X_j \chi_L= X_j \chi_L$  is bounded and, also $[P,X_j]$ is bounded by Proposition~\ref{prop:[nearsighted,Xj] is bounded} and Lemma~\ref{lem:P is nearsighted}, we can deduce from the expression above that $\chi_L\Sigma_K^{s_z}\chi_L$ is trace class for every $L\in 2\N+1$.

\item[Part \ref{item:KG12s=Ksigma12s}:] We introduce the function 
\begin{equation}
\label{eqn:linear Xi}
\Xi(n_1)=
\begin{cases}
0 & \text{if $n_1 < -1/2$,} \\
n_1+1/2 & \text{if $-1/2 \leq n_1 < 1/2$} \\
1 & \text{if $n_1\geq 1/2$},
\end{cases}
\end{equation}
which interpolates linearly in the interval $\abs{n_1}\leq 1/2$ and, for $l>0$ we define the functions $\Xi^{(l)}(n_1):=\Xi(\frac{n_1}{l})$ which have slope $1/l$ in the interval $\abs{n_1}\leq l/2$. Now, we define the \crucial{approximate position functions} in the $1^{\mathrm{st}}$-direction as
\begin{equation}
\label{eqn:defn appr pos funct}
X_1^{(l)}:= l\( \Xi^{(l)} - \frac{1}{2} \)\,\text{ such that }\, X_1^{(l)}(n_1)=
\begin{cases}
-l/2 & \text{if $n_1 < -l/2$,} \\
\phantom{-}n_1 & \text{if $-l/2 \leq n_1 < l/2$} \\
\phantom{-}l/2 & \text{if $n_1\geq l/2$}.
\end{cases}
\end{equation}
Notice that for every $l>0$ the functions $\Xi^{(l)}$ are particular switch functions in the $1^{\mathrm{st}}$-direction. 

We now compute $G_K^{s_z}(\Xi^{(l)}, \Lambda_2)$ and show that it is finite. 
In view of Part \ref{item:defn GKs}, this fact will imply that  $G_K^{s_z}(\Lambda_1, \Lambda_2)$ is finite for every switch function 
$\Lambda_1$, and independent of the choice of the latter.  Notice that 

\begin{align}
\label{eqn:KG12s=Ksigma12s step1}
%G_K^{s_z}( \Lambda_1, \Lambda_2)&=
G_K^{s_z}\big(\Xi^{(l)},\Lambda_2\big)&=
G_K^{s_z}\big(\Xi^{(l)}-\half,\Lambda_2\big)+G_K^{s_z}\big(\half,\Lambda_2\big)\cr
&=\frac{1}{l}G_K^{s_z}\big(X_1^{(l)},\Lambda_2\big)+\frac{1}{2}G_K^{s_z}\big(\Id,\Lambda_2\big),
\end{align}
provided the two summands separately exist and are finite (which is what we are going to prove).

We focus attention on the first summand on the right-hand side of the last equation. 
Recall that, by definition \eqref{Eq:G_K}, one has
$$
G_K^{s_z}\big( X_1^{(l)},\Lambda_2 \big)=1\text{-}\pvTr\big( \G_K^{s_z}\big(X_1^{(l)},\Lambda_2\big) \big)
$$
where 
$$
\G_K^{s_z}\big(X_1^{(l)},\Lambda_2\big)=\ii P [[P, X_1^{(l)} S_z], \, [P,\Lambda_2]]P.
$$
We analyse $\G_K^{s_z}\big(X_1^{(l)},\Lambda_2\big)$.
By algebraic manipulations, using $P^2=P$ and $P^\perp = \Id-P$, in view of Leibniz rule for the product $X_1^{(l)} S_z$ and $[X_1^{(l)},S_z]=0$, %\footnote{This elementary remark will be very helpful.} 
we obtain
\begin{align}
\label{eqn:KG12s=Ksigma12s step2}
\G_K^{s_z}\big(X_1^{(l)},\Lambda_2\big) &=\ii [P, X_1^{(l)} S_z] P^\perp [P,\Lambda_2] - \ii [P, \Lambda_2 ] P^\perp [P,X_1^{(l)} S_z]\cr
&=\underbrace{\ii [P, X_1^{(l)}] S_z P^\perp [P,\Lambda_2]}_{=: \G_{K,a}^{s_z}}+\underbrace{- \ii [P, \Lambda_2 ] P^\perp S_z [P,X_1^{(l)} ]}_{= {\(\G_{K,a}^{s_z}\)}^*}\cr
&\phantom{=}+ \underbrace{X_1^{(l)} \ii[P,  S_z] P^\perp [P,\Lambda_2]}_{=: \G_{K,b}^{s_z}}+\underbrace{- \ii [P, \Lambda_2 ] P^\perp  [P, S_z] X_1^{(l)}}_{= {\(\G_{K,b}^{s_z}\)}^*}.
\end{align}
Notice that $\G_{K,a}^{s_z}\big(X_1^{(l)},\Lambda_2\big)=\ii [P, X_1^{(l)}] S_z  [P,\Lambda_2]P$ is trace class. As $P$ is bounded, it is enough to prove that $[P, X_1^{(l)}] S_z  [P,\Lambda_2]$ is trace class. By Proposition~\ref{prop:suff traceclass cond ABC} applied to $A=[P, X_1^{(l)}]=l[P,\Xi^{(l)}]$, which is $\alpha$-confined in the $1^{\mathrm{st}}$-direction for some $\alpha<1/\zeta_P$ by Lemma~\ref{lem:[nearsighted,switch] is conf} and Lemma~\ref{lem:P is nearsighted}, $B^*=[\Lambda_2,P]$ is $\beta$-confined in the $2^{\mathrm{nd}}$-direction for some $\beta<1/\zeta_P$ by Lemma~\ref{lem:[nearsighted,switch] is conf} and Lemma~\ref{lem:P is nearsighted}, and $C=S_z$, we have the trace class property for $\G_{K,a}^{s_z}\big(X_1^{(l)},\Lambda_2\big)$. Therefore, as $\TC$ is closed under adjointness, we also have ${\G_{K,a}^{s_z}}^*\big(X_1^{(l)},\Lambda_2\big)\in\TC$.

Therefore, in view of Proposition~\ref{prop:jpv tr}, we have
\begin{equation}
1\text{-}\pvTr\big(\G_{K,a}^{s_z}\big(X_1^{(l)},\Lambda_2\big)+\adj \big)=\Tr\big(\G_{K,a}^{s_z}\big(X_1^{(l)},\Lambda_2\big)+\adj \big)
\end{equation}
which is finite. As explained in Appendix \ref{app:Switch}, for periodic operators the trace of an expression involving switch functions may become a trace on the unit cell where position operators replace commutators with switch functions. In particular, by Lemma~\ref{lem:loc12} we deduce
\begin{align*}
\frac{1}{l}\Tr\big(\G_{K,a}^{s_z}\big(X_1^{(l)},\Lambda_2\big)+\adj \big)&=\frac{1}{l}\Tr\big( \ii [P, X_1^{(l)}] S_z P^\perp [P,\Lambda_2]   +\adj \big)\cr
&=\Tr\big( \ii [P,\Xi^{(l)}-\half ] S_z P^\perp [P,\Lambda_2]   +\adj \big)\cr
&=\Tr\big( \ii [P,\Xi^{(l)}] S_z P^\perp [P,\Lambda_2]   +\adj \big)\cr
&= \Tr (-\chi_1 \iu P X_1 S_z P^\perp X_2 P \chi_1 +\adj)\cr
&= \Tr (\chi_1 \iu [P ,X_1 S_z]P^\perp[P, X_2]\chi_1 +\adj )\cr
&=\Tr (\chi_1 \iu [P ,X_1 S_z]P^\perp[P, X_2]\chi_1 -\chi_1 \iu[P, X_2]P^\perp [P ,X_1 S_z]   \chi_1)\cr
&=\Tr (\chi_1 \iu P [[P ,X_1 S_z],[P, X_2]] P \chi_1).
\end{align*}
Finally, by Part~\ref{item:defn sigmaKs} and by the last equation we obtain that%%%
\footnote{Notice that we do not need to consider the limit $l \rightarrow + \infty$, as one might expect.}\  %%%%

\begin{equation}
\label{eqn:KG12s=Ksigma12s step3}
\frac{1}{l}\Tr\big(\G_{K,a}^{s_z}\big(X_1^{(l)},\Lambda_2\big)+\adj \big)=\tau( \iu P [[P ,X_1 S_z],[P, X_2]] P )=\sigma_K^{s_z}.
\end{equation}

Now, we compute $1\text{-}\pvTr\big(\G_{K,b}^{s_z}+\adj\big)$, whose argument is defined in the equation~\eqref{eqn:KG12s=Ksigma12s step2}. 

Notice that $\chi_{1,L}\G_{K,b}^{s_z}\chi_{1,L}=\chi_{1,L} X_1^{(l)} \ii[P,  S_z] P^\perp [P,\Lambda_2]\chi_{1,L}$ is trace class, as it follows from Proposition~\ref{prop:suff traceclass cond ABC} with $A=[P,\Lambda_2]$, which is $\alpha$-confined in the $2^{\mathrm{nd}}$-direction for some $\alpha<1/\zeta_P$ by Lemma~\ref{lem:[nearsighted,switch] is conf} and Lemma~\ref{lem:P is nearsighted}, $B=\chi_{1,L}$, which is $\beta$-confined in the $1^{\mathrm{st}}$-direction for some $\beta>0$, and $C=\Id$. As $\TC$ is closed under adjointness, we have that ${\big(\chi_{1,L}\G_{K,b}^{s_z}\chi_{1,L}\big)}^*$ is also trace class. By Lemma~\ref{lem:loc2}, we obtain
\begin{align}
\label{eqn:KG12s=Ksigma12s step4}
\Tr\big(\chi_{1,L}\G_{K,b}^{s_z}\chi_{1,L}+\adj\big)&=\Tr\big(\chi_{1,L} X_1^{(l)} \ii[P,  S_z] P^\perp [P,\Lambda_2]\chi_{1,L} +\adj\big)\cr
&=\Tr\big(-\chi_{1,L} X_1^{(l)}\ii[P,  S_z]P^\perp X_2 P\chi_{2,1}\chi_{1,L} +\adj\big)\cr
&=-\Tr\big(\chi_{1,L} X_1^{(l)}\ii[P,  S_z]P^\perp X_2 P\chi_{2,1}\chi_{1,L}\big) +\cr &\phantom{=\;}-\Tr\big(\chi_{1,L}\chi_{2,1}PX_2P^\perp\ii[P,  S_z]X_1^{(l)}\chi_{1,L}\big)\cr
&=\Tr\big(\chi_{1,L} X_1^{(l)}\ii[P,  S_z]P^\perp [P,X_2 ]\chi_{2,1}\chi_{1,L}\big)+\cr
&\phantom{=\;}-\Tr\big(\chi_{1,L}\chi_{2,1}[P,X_2]P^\perp\ii[P,  S_z]X_1^{(l)}\chi_{1,L}\big).
\end{align}
Notice that the operator $\chi_{1,L} X_1^{(l)}\ii[P,  S_z]P^\perp [P,X_2 ]\chi_{2,1}\chi_{1,L}$ is trace class, because $\chi_{2,1}\chi_{1,L}$ is trace class applying Proposition~\ref{prop:suff traceclass cond ABC} where $A=\chi_{2,1}$, which is $\alpha$-confined in the $2^{\mathrm{nd}}$-direction for some $\alpha>0$,  $B=\chi_{1,L}$, which is $\beta$-confined in the $1^{\mathrm{st}}$-direction for some $\beta>0$, and $C=\Id$, and $[P,X_2 ]$ is bounded by Proposition~\ref{prop:[nearsighted,Xj] is bounded} and Lemma~\ref{lem:P is nearsighted}, $\chi_{1,L} X_1^{(l)}$ and $\ii[P,  S_z]P^\perp$ are also bounded. Thus, as $\chi_{2,1}$ squares to itself, by using Proposition~\ref{prop:cycl of trace} we obtain 
\begin{equation}
\label{eqn:KG12s=Ksigma12s step5}
\Tr\big(\chi_{1,L} X_1^{(l)}\ii[P,  S_z]P^\perp [P,X_2 ]\chi_{2,1}\chi_{1,L}\big)=\Tr\big(\chi_{1,L}\chi_{2,1} X_1^{(l)}\ii[P,  S_z]P^\perp [P,X_2 ]\chi_{2,1}\chi_{1,L}\big).
\end{equation}
Similarly, using also that multiplicative operators by position functions commute, we obtain
\begin{equation}
\label{eqn:KG12s=Ksigma12s step6}
\Tr\big(\chi_{1,L}\chi_{2,1}[P,X_2]P^\perp\ii[P,  S_z]X_1^{(l)}\chi_{1,L}\big)=\Tr\big(\chi_{1,L}\chi_{2,1}X_1^{(l)}[P,X_2]P^\perp\ii[P,  S_z]\chi_{2,1}\chi_{1,L}\big).
\end{equation}
Therefore, plugging \eqref{eqn:KG12s=Ksigma12s step5} and \eqref{eqn:KG12s=Ksigma12s step6} into the equation \eqref{eqn:KG12s=Ksigma12s step4}, we have
\begin{align*}
\Tr\big(\chi_{1,L}\G_{K,b}^{s_z}\chi_{1,L}+\adj\big)&=\Tr\big(\chi_{1,L}\chi_{2,1} X_1^{(l)}\big(\ii[P,  S_z]P^\perp [P,X_2 ]-[P,X_2]P^\perp\ii[P,  S_z]\big)\chi_{2,1}\chi_{1,L}\big)\cr
&=\Tr\big(\chi_{1,L}\chi_{2,1} X_1^{(l)}\ii P[[P,  S_z],[P,X_2]]P\chi_{2,1}\chi_{1,L}\big)\cr
&=\Tr\big(\chi_{1,L}\chi_{2,1}X_1^{(l)}\mathcal{T}_{s_z}\chi_{2,1}\chi_{1,L}\big).
\end{align*}
Observe that for every fixed $L\in 2\N+1$, the operator $\chi_{1,L}\chi_{2,1}X_1^{(l)}\mathcal{T}_{s_z}\chi_{2,1}\chi_{1,L}$ is trace class, as $\chi_{1,L}\chi_{2,1}$ is trace class for the previous analysis and $\mathcal{T}_{s_z}$ is bounded by Lemma~\ref{lem:spin torque}. We compute its trace through the diagonal kernel, using Lemma~\ref{lem:spin torque},
\begin{align}
\label{eqn:KG12s=Ksigma12s step12}
\Tr\big(\chi_{1,L}\chi_{2,1}X_1^{(l)}\mathcal{T}_{s_z}\chi_{2,1}\chi_{1,L}\big)&=\sum_{\substack{m_1\in\Z\\\abs{m_1}\leq L/2}}X_1^{(l)}(m_1)\tr\big({(\mathcal{T}_{s_z})}_{(m_1,0),(m_1,0)}\big)\cr
&=\sum_{\substack{m_1\in\Z\\\abs{m_1}\leq L/2}}X_1^{(l)}(m_1)\tr\big({(\mathcal{T}_{s_z})}_{\V{0},\V{0}}\big)\cr
&=\Tr(\chi_{1}\mathcal{T}_{s_z}\chi_{1})\sum_{\substack{m_1\in\Z\\\abs{m_1}\leq L/2}}X_1^{(l)}(m_1)\cr
&=\tau(\mathcal{T}_{s_z})\sum_{\substack{m_1\in\Z\\\abs{m_1}\leq L/2}}X_1^{(l)}(m_1)\equiv 0,
\end{align}
as the function $X_1^{(l)}(m_1)$ is odd and the interval $\abs{m_1}\leq L/2$ is symmetric with respect to $0$.
(We could also invoke the fact that $\tau(\mathcal{T}_{s_z})=0$ by Theorem \ref{thm:main1}).
Thus,
\begin{equation}
\label{eqn:KG12s=Ksigma12s step7}
1\text{-}\pvTr\big(\G_{K,b}^{s_z}+\adj\big)=0.
\end{equation}

Using equations \eqref{eqn:KG12s=Ksigma12s step3} and \eqref{eqn:KG12s=Ksigma12s step7}, we obtain 
\begin{equation}
\label{eqn:KG12s=Ksigma12s step11}
\frac{1}{l}1\text{-}\pvTr\big(\G_{K}^{s_z}\big(X_1^{(l)},\Lambda_2\big) \big)=\sigma_K^{s_z}+0=\sigma_K^{s_z}.
\end{equation}

Now, we focus attention on the second summand on the right-hand side of \eqref{eqn:KG12s=Ksigma12s step1}.
We have
\begin{align}
\label{eqn:KG12s=Ksigma12s step8}
\frac{1}{2}G_K^{s_z}\big(\Id,\Lambda_2\big)&=\frac{1}{2} 1\text{-}\pvTr\big(\G_K^{s_z}\big(\Id,\Lambda_2\big) \big)\cr
&=\frac{1}{2}\lim_{\substack{L\to\infty \\ L\in 2\N+1}}\Tr(\chi_{1,L}\ii P [[P, S_z], \, [P,\Lambda_2]]P \chi_{1,L}).
\end{align}
Notice that $\chi_{1,L}\ii P [[P, S_z], \, [P,\Lambda_2]]P \chi_{1,L}$ is trace class, as one proves by applying Proposition~\ref{prop:suff traceclass cond ABC} and reasoning as in the previous cases.

By Lemma~\ref{lem:loc2}, the identity $\chi_{2,1}^2=\chi_{2,1}$ and Proposition~\ref{prop:cycl of trace}, we obtain
\begin{align}
\label{eqn:KG12s=Ksigma12s step9}
\Tr(\chi_{1,L}\ii P [[P, S_z], \, [P,\Lambda_2]]P \chi_{1,L})&=\Tr(\chi_{1,L}\ii [P, S_z] P^\perp [P,\Lambda_2] \chi_{1,L}+\adj)\cr
&=\Tr(-\chi_{1,L}\ii [P, S_z] P^\perp X_2 P\chi_{2,1} \chi_{1,L}+\adj)\cr
&=-\Tr(\chi_{1,L}\ii [P, S_z] P^\perp X_2 P\chi_{2,1} \chi_{1,L})+\cr
&\phantom{=\;}-\Tr (\chi_{1,L}\chi_{2,1} P X_2P^\perp \ii [P, S_z] \chi_{1,L} )\cr
&=\Tr(\chi_{1,L}\ii [P, S_z] P^\perp [P,X_2] \chi_{2,1} \chi_{1,L})+\cr
&\phantom{=\;}-\Tr(\chi_{1,L}\chi_{2,1} [P, X_2]P^\perp \ii [P, S_z] \chi_{1,L})\cr
&=\Tr(\chi_{1,L} \chi_{2,1}\mathcal{T}_{s_z} \chi_{2,1}\chi_{1,L}).
\end{align}
As $\chi_{1,L} \chi_{2,1}\mathcal{T}_{s_z} \chi_{2,1}\chi_{1,L}$ is trace class, computing its trace via diagonal kernel and using Lemma~\ref{lem:spin torque}, we get
\begin{align*}
\Tr(\chi_{1,L} \chi_{2,1}\mathcal{T}_{s_z} \chi_{2,1}\chi_{1,L})&=\sum_{\substack{m_1\in\Z\\ \abs{m_1}\leq L/2}}\tr({(\mathcal{T}_{s_z})}_{(m_1,0),(m_1,0)})\cr
&=\sum_{\substack{m_1\in\Z\\ \abs{m_1}\leq L/2}}\tr({(\mathcal{T}_{s_z})}_{\V{0},\V{0}})
=\sum_{\substack{m_1\in\Z\\ \abs{m_1}\leq L/2}}\tau(\mathcal{T}_{s_z}).
\end{align*}
Thus, plugging the last equality and equation~\eqref{eqn:KG12s=Ksigma12s step9} in \eqref{eqn:KG12s=Ksigma12s step8}, we obtain
\begin{align}
\label{eqn:KG12s=Ksigma12s step10}
\frac{1}{2}G_K^{s_z}\big(\Id,\Lambda_2\big)=\frac{1}{2}\lim_{\substack{L\to\infty \\ L\in 2\N+1}}\sum_{\substack{m_1\in\Z\\ \abs{m_1}\leq L/2}}\tau(\mathcal{T}_{s_z}) = 0
\end{align}
in view of Theorem~\ref{thm:main1}. This concludes the proof of Theorem \ref{thm:main2}. \qed

\bigskip \bigskip

It is worthwhile to notice that, without using Theorem~\ref{thm:main1}, by
plugging equalities \eqref{eqn:KG12s=Ksigma12s step11} and \eqref{eqn:KG12s=Ksigma12s step10} into  \eqref{eqn:KG12s=Ksigma12s step1}, one would obtain 
\begin{align}
\label{eqn:KG12s=Ksigma12s general proof} 
G_K^{s_z}( \Lambda_1, \Lambda_2)&=\frac{1}{l}G_K^{s_z}\big(X_1^{(l)},\Lambda_2\big)+\frac{1}{2}G_K^{s_z}\big(\Id,\Lambda_2\big)\cr
&=\sigma_K^{s_z}+\frac{1}{2}\lim_{\substack{L\to\infty \\ L\in 2\N+1}}\sum_{\substack{m_1\in\Z\\ \abs{m_1}\leq L/2}}\tau(\mathcal{T}_{s_z}).
\end{align}
As remarked in Section 2, the second summand on the right hand side is either zero, if $\tau(\mathcal{T}_{s_z})=0$, or diverging to $\pm \infty$.  Hence, the equality of (the Kubo-like terms of) the spin conductance and spin conductivity  is rooted in the fact that the spin-torque response $\tau(\mathcal{T}_{s_z})$ vanishes on the mesoscopic scale. We expect that such a physically relevant condition will play a role also in other models, as \eg ergodic random Schr\"odinger operators.
\end{itemize}

\goodbreak
\newpage

\appendix

%%%%% APPENDIX 1  %%%%%%%%%%%%%%%%%%%%%%%%%%%%%%%%%%%%

\section[The Kane-Mele model ]{The Kane-Mele model in first quantization formalism}
\label{Sec:Kane-Mele}

In this appendix we review an explicit model that satisfies Assumption \ref{ass:H} and where spin is not conserved.  This model was first introduced by Kane and Mele in \cite{KaneMele2005,KaneMele_graphene}. Here we propose a first quantized formulation of it, but first we discuss the dimerization method mentioned at the beginning of Section \ref{sec:math setting and main result}. 

%%%%%%%%%%%%%%%%%%%%%%%%%%%%%%%%%%%%%%%%%
\medskip

\noindent \textbf{A.1. The honeycomb structure.} 
The model describes independent electrons on a honeycomb structure $\mathcal{C}$, illustrated in Figure \ref{fig:honeycomb}. 
The structure is characterized by the \emph{displacement vectors} 
$$
\V d_1 = d\begin{pmatrix} \frac{1}{2} & -\frac{\sqrt{3}}{2} \end{pmatrix}, \qquad \V d_2 = d\begin{pmatrix} \frac{1}{2} &\frac{\sqrt{3}}{2} \end{pmatrix}, \qquad \V d_3 = d\begin{pmatrix} -1 & 0 \end{pmatrix} = -\V d_1-\V d_2,
$$
where $d$ is the {smallest} distance between two points of $\mathcal{C}$, which generate the \emph{periodicity vectors} 
\begin{equation}
\V a_1  = \V d_2 - \V d_3, \qquad \V a_2 = \V d_3-\V d_1, \qquad \V a_3 = \V d_1-\V d_2 = -\V a_1 - \V a_2.
\end{equation}

\begin{figure}[htb]
\centering
\includegraphics{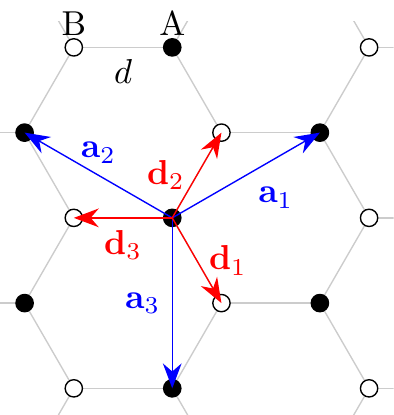}
\caption{The honeycomb structure. \label{fig:honeycomb}}
\end{figure}

\begin{figure}[htb]
\centering
\includegraphics[scale=0.9]{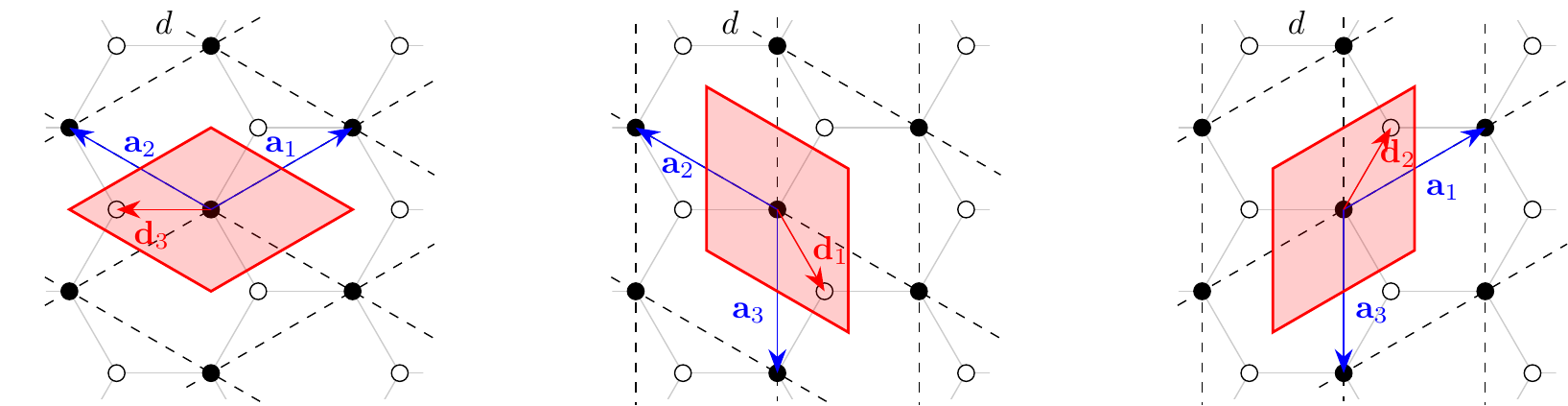}
\caption{Three possible dimerizations of the honeycomb structure. \label{fig:unitcells}}
\end{figure}

The vectors $\V a_i$ generate a Bravais lattice $\Gamma := \mathrm{Span}_\mathbb{Z}\{ \V a_1, \V a_2, \V a_3\} \cong \mathbb Z^2$ where one $\V a_i$ is redundant as it is integer linear combination of the two others. Then any site of the crystal can be reached by a Bravais lattice vector and the use of one of the $\V d_i$ vectors. It is then sufficient to pick two $\V a_i$-vectors and one $\V d_i$-vector to generate  the whole crystal. This choice,  which is often called a \emph{dimerization} of $\mathcal C$, is not unique, as illustrated in Figure \ref{fig:unitcells}.

The above procedure is equivalent to the choice of a periodicity cell that contains two-non equivalent sites $A$ and $B$ (white and black dots in Figure \ref{fig:honeycomb}), described as internal degrees of freedom besides the Bravais lattice. 
Hence, each choice of unit cell provides an isomorphism $\ell^2(\mathcal C) \cong \ell^2(\mathbb Z^2) \otimes \mathbb C^2$, leading to the Hilbert space $\mathcal \Hi$ (for $N=2$)  discussed in Section \ref{sec:math setting and main result}, when the spin is taken into account.

%%%%%%%%%%%%%%%%%%%%%%%%%%%%%%%%%%%%%%%%%
\medskip

\noindent \textbf{A.2. The  Hamiltonian.}
The Kane-Mele model is defined, in a first quantization formalism, 
by the Hamiltonian $H\sub{KM}$, acting on  $\ell^2(\mathcal C) \otimes \mathbb C^2$ as
$$
H_\mathrm{KM} = t H_\mathrm{NN} + \lambda_v H_v + \lambda_\mathrm{SO} H_\mathrm{SO} + \lambda_\mathrm{R} H_\mathrm{R}
$$
where $t,\,\lambda_v,\lambda_\mathrm{SO}$ and $\lambda_\mathrm{R}$ are real parameters corresponding to various physical effects. The first term is a nearest neighbor hopping term:
$$
H_\mathrm{NN} = \sum_{i=1}^3 \left( T_{\V d_i} + T_{-\V d_i} \right) \otimes \Id_{\mathbb C^2}
$$
where $T_{\V u}$ is a translation operator along vector $\V u$, namely
\begin{equation*} 
\label{e}
(T_{\V u} \psi)_{\V m} = \begin{cases}
\psi_{\V m-\V u}             & \text{ if }  {\V m-\V u} \in \mathcal{C};   \\
0                                    & \text{ otherwise. }
\end{cases}
\end{equation*}
The second term is a sublattice potential that distinguishes sites $A$ and $B$, namely
$$
H_v = (\chi_{A} - \chi_{B})\otimes \Id_{\mathbb C^2}
$$
for $\chi_A$ (resp. $\chi_B$) the characteristic function on the sublattice $A$ (resp. sublattice $B$) of $\mathcal C$. The third term is a spin-orbit term, corresponding to an effective and spin-dependent magnetic field due to an electric field inside the two-dimensional crystal. This is a next-to-nearest-neighbor term given by
$$
H\sub{SO} = - \ii \left( \chi_{A} - \chi_{B}\right) \sum_{i=1}^3 \left( T_{\V a_i} - T_{-\V a_i} \right) \otimes s_z.
$$
Finally the last term is called a Rashba term. This is also a spin-orbit effect but due to an electric field orthogonal to the sample (for example in a heterostructure). This is a nearest-neighbor term given by
$$
H_{\mathrm R} = \ii \big( T_{\V d_1} -  T_{-\V d_1} \big) \otimes  \Big(- \dfrac{\sqrt 3 s_x + s_y }{2} \Big) + \ii \big( T_{\V d_2} -  T_{-\V d_2} \big) \otimes  \Big( \dfrac{\sqrt 3 s_x - s_y }{2} \Big)  + \ii \big( T_{\V d_3} -  T_{-\V d_3} \big) \otimes s_y 
$$ 

Notice that this last term satisfies $[H_\mathrm{R},S_z] \neq 0$ so that $S_z$ and $H_\mathrm{KM}$ do not commute whenever $\lambda_R \neq 0$. Moreover, note that $H_\mathrm{KM}$ is periodic, since $[T_{\V u_1},T_{\V u_2}] =0$ for any vectors $\V u_1$ and $\V u_2$. In particular $H_\mathrm{KM}$ commutes with all the translation of the Bravais lattice $T_\V \gamma$ for $\V \gamma \in \Gamma$. It was also shown in \cite{KaneMele2005} that $H_\mathrm{KM}$ has a spectral gap for a wide region in parameter space, including $\lambda_R \neq 0$ (Figure 1 in \cite{KaneMele2005}). 

In summary, $H_\mathrm{KM}$ is made of on-site ($H_v$), nearest-neighbor ($H_\mathrm{NN}$ and $H_\mathrm{R}$) and next to nearest-neighbor ($H_\mathrm{SO}$) terms. Note that after the dimerization procedure a nearest-neighbor term acts on internal degree of freedom,  whereas next-to-nearest-neighbor exchange becomes simply nearest-neighbor. Thus, whatever the dimerization, one has
$$
(H_\mathrm{KM})_{\V{m},\V{n}} = 0 \qquad \text{for} \qquad  \norm{\V{m}-\V{n}}_1 > 1
$$
so that $H_\mathrm{KM}$ is trivially near-sighted. Indeed by adapting $C$ the inequality of Definition \ref{defn:op near-sighted},  $H_\mathrm{KM}$ is  near-sighted for any range $\zeta >0$.

\goodbreak

%%%%% APPENDIX 2  %%%%%%%%%%%%%%%%%%%%%%%%%%%%%%%%%%%%

\section{From switch functions to position operators}
\label{app:Switch}
In this Appendix, we re-elaborate some ideas and techniques which originally appeared in 
 \cite{AvronSeilerSimon} in the continuum case ($\R^2$-covariant Schr\"odinger operators on the plane). We adapt their proof to the discrete case considered in this paper.

The crucial property of any switch function is the following one.
\begin{lemma}
\label{lem:parallelogram area}
Let $\Lambda_j$ be a switch function in the $j^{\mathrm{th}}$-direction for $j\in\set{1,2}$. Then, for every $n\in\Z$ one has
$$
\sum_{m \in \Z} \( \Lambda_j(m+n)-\Lambda_j(m) \) = n.
$$
\end{lemma}
\begin{proof}
For $n=0$ the claim is trivial. Consider $n\geq 1$.  Notice that the summand $\Lambda_j(m+n)-\Lambda_j(m)$ is non-zero only for finitely many $m \in \Z$. Hence, 
%$$
%\sum_{p=0}^{n-1}\sum_{m\in\Z}\abs{\Lambda_j(m)-\Lambda_j(m-1)}<\infty
%$$ 
%by Fubini's Theorem
\begin{align}
\label{eqn:discrete FTIC}
\sum_{m\in\Z} \( \Lambda_j(m+n)-\Lambda_j(m) \) 
&=\sum_{m\in\Z}\sum_{p=0}^{n-1}  \big( \Lambda_j(m+(n-p))-\Lambda_j(m+(n-p-1)) \big) \cr
%&=\sum_{p=0}^{n-1}\sum_{m\in\Z}\Lambda_j(m+(n-p))-\Lambda_j(m+(n-p-1))\cr
&=\sum_{p=0}^{n-1}\sum_{m\in\Z} \big( \Lambda_j(m)-\Lambda_j(m-1) \big).
\end{align} 
Notice that $\sum_{m\in\Z} \big( \Lambda_j(m)-\Lambda_j(m-1) \big) =1$, since there is one and only one point $m \in \Z$ where the summand is not zero. This proves the statement for $n \geq 1$. 
The proof for $n \leq -1$ is analogous.  
%\begin{align}
%\label{eqn:diff Lambda and 1shifted Lambda}
%\sum_{m\in\Z}\Lambda_j(m)-\Lambda_j(m-1)&=\sum_{m\in\Z}\Bigl(\chi_{\set{m\geq n_+}}+\Lambda_j(m)\chi_{\set{n_-\leq m< n_+}}+\cr
%&\phantom{\sum_{m\in\Z}\bigg(\;}-\chi_{\set{m-1\geq n_+}}-\Lambda_j(m-1)\chi_{\set{n_-\leq m-1< n_+}}\Bigr)\cr
%&=\sum_{m\in\Z}\Bigl(\chi_{\set{m\geq n_+}}-\chi_{\set{m-1\geq n_+}}\Bigr)+\cr
%&\phantom{=}+\sum_{\substack{m\in\Z\\n_-\leq m< n_+}}\Lambda_j(m)-\sum_{\substack{m\in\Z\\n_-\leq m-1< n_+}}\Lambda_j(m-1)\cr
%&=1,
%\end{align}
%where we have used that in the second equality we can split the series as sum of three series since the two last ones are finite.
%Plugging the equality~\eqref{eqn:diff Lambda and 1shifted Lambda} into \eqref{eqn:discrete FTIC}, we obtain the claim.
%
%If $n=-|n|$ such that $|n|\in\N$ with $|n|\geq 1$, then 
%\begin{align*}
%\sum_{m\in\Z} \Lambda_j(m+n)-\Lambda_j(m)&=\sum_{m\in\Z} \Lambda_j(m-|n|)-\Lambda_j(m)=\sum_{m\in\Z} \Lambda_j(m)-\Lambda_j(m+|n|)\cr
%&=-\sum_{m\in\Z} \left(\Lambda_j(m+|n|)-\Lambda_j(m)\right)=-\abs{n}=n,
%\end{align*}
%where we have used the statement for $|n|\in\N$.
\end{proof}

For the sake of clarity, we recall that $\chi_{2,1}$ and $\chi_{1}$ are characteristic functions, respectively of the line 
$\set{\V{m}\in\Z^2: m_2=0}$ and of the point $\set{\V{0}}$.

\begin{lemma}
	\label{lem:loc2}
	Let $A$, $B$ and $C$ be operators in $\BHi$ which are periodic in the $2^{\mathrm{nd}}$-direction and let $\Lambda_2$ be a switch function in the $2^{\mathrm{nd}}$-direction. 
	If $A[B,\Lambda_2]C$ is trace class, $A$ is $\alpha$-confined in the $1^{\mathrm{st}}$-direction, $C^*$ is $\beta$-confined in the $1^{\mathrm{st}}$-direction and $B$ satisfies 
	\begin{equation}
	\label{eqn:hyp B loc2}
	\mathcal{M}_B:=\max \(\sup_{\V{m}\in\Z^2}\sum_{n_1\in\Z}\abs{ B_{\V{m},(n_1,0)}m_2}, \sup_{n_1\in\Z}\sum_{\V{m}\in\Z^2}\abs{ B_{\V{m},(n_1,0)}m_2} \)<\infty,
	\end{equation}
	then
	$$
	\Tr(A[B,\Lambda_2]C) =- \Tr( A X_2 B\chi_{2,1}C).
	$$
\end{lemma}
\begin{proof}
	Since $A[B,\Lambda_2]C$ is trace class, its trace can be computed through the diagonal kernel, and in view of the boundedness of $A$, $[B,\Lambda_2]$ and $C$, one has
	\begin{equation}
	\label{eqn:loc2 step1}
	\Tr(A[B,\Lambda_2] C)=\sum_{\V{m} \in \Z^2}\sum_{\V{n} \in \Z^2}\sum_{\V{p}\in\Z^2}\tr\big(A_{\V{m}, \V{n}}B_{\V{n}, \V{p}}(\Lambda_2(p_2)-\Lambda_2(n_2))C_{\V{p},\V{m}} \big).
	\end{equation}
	Now, notice that the function 
	 \begin{equation} \label{is in ell^1}
         {\big(\Z^2\big)}^3\ni (\V{m,n,p})\mapsto\tr\big(A_{\V{m}, \V{n}}B_{\V{n}, \V{p}}(\Lambda_2(p_2)-\Lambda_2(n_2))C_{\V{p}, 
         \V{m}} \big)\text{ is in } \ell^1\big({\big(\Z^2\big)}^3\big).
         \end{equation}
	Indeed, in view of the equivalence of norms on finite-dimensional vector spaces and the periodicity in the $2^{\mathrm{nd}}$-direction, first one notices that
	\begin{align*}
	&\sum_{\V{m} \in \Z^2}\sum_{\V{n} \in \Z^2}\sum_{\V{p}\in\Z^2}\abs{\tr\big(A_{\V{m}, \V{n}}B_{\V{n}, \V{p}}(\Lambda_2(p_2)-\Lambda_2(n_2))C_{\V{p},\V{m}} \big)}\cr
	&\leq D_1\sum_{\V{m} \in \Z^2}\sum_{\V{n} \in \Z^2}\sum_{\V{p}\in\Z^2}\abs{A_{\V{m}, \V{n}}B_{\V{n}, \V{p}}(\Lambda_2(p_2)-\Lambda_2(n_2))C_{\V{p},\V{m}} }\cr
	&\leq \frac{D_1}{2}\sum_{\V{m} \in \Z^2}\sum_{\V{n} \in \Z^2}\sum_{\V{p}\in\Z^2}\abs{B_{\V{n}, \V{p}}}\abs{(\Lambda_2(p_2)-\Lambda_2(n_2))}\big(\abs{A_{\V{m}, \V{n}}}^2+\abs{C_{\V{p},\V{m}} }^2\big)\cr
	&\leq \frac{D_1}{2}\sum_{\V{n} \in \Z^2}\sum_{\V{m}' \in \Z^2}\sum_{\V{p}'\in\Z^2}\abs{B_{\V{n}, \V{p}'+(0,n_2)}}\abs{(\Lambda_2(p_2'+n_2)-\Lambda_2(n_2))}\abs{A_{\V{m}'+(0,n_2), \V{n}}}^2+\cr
	&\phantom{\leq}+\frac{D_1}{2}\sum_{\V{p}\in\Z^2}\sum_{\V{m}' \in \Z^2}\sum_{\V{n}' \in \Z^2}\abs{B_{\V{n}'+(0,p_2), \V{p}}}\abs{(\Lambda_2(p_2)-\Lambda_2(n_2'+p_2))}\abs{C_{\V{p},\V{m}'+(0,p_2)} }^2\cr
	&\leq \frac{D_1}{2}\sum_{\V{n} \in \Z^2}\sum_{\V{m}' \in \Z^2}\sum_{\V{p}'\in\Z^2}\abs{B_{(n_1,0), \V{p}'}}\abs{(\Lambda_2(p_2'+n_2)-\Lambda_2(n_2))}\abs{A_{\V{m}', (n_1,0)}}^2+\cr
	&\phantom{\leq}+\frac{D_1}{2}\sum_{\V{p}\in\Z^2}\sum_{\V{m}' \in \Z^2}\sum_{\V{n}' \in \Z^2}\abs{B_{\V{n}', (p_1,0)}}\abs{(\Lambda_2(p_2)-\Lambda_2(n_2'+p_2))}\abs{C_{(p_1,0),\V{m}'} }^2,
	\end{align*}
	where $D_1$ is a constant.
	Then, one can estimates the right-hand side term of the last equation from above with 
	\begin{align*}
	& D_1 D_2(\Lambda_2)\Big(\sum_{n_1 \in \Z}\sum_{\V{m}' \in \Z^2}\sum_{\V{p}'\in\Z^2}\abs{B_{(n_1,0), \V{p}'}}\abs{p_2'}\abs{A_{\V{m}', (n_1,0)}}^2+\cr
	&\phantom{D_1 D_2(\Lambda_2)\Big(} +\sum_{p_1\in\Z}\sum_{\V{m}' \in \Z^2}\sum_{\V{n}' \in \Z^2}\abs{B_{\V{n}', (p_1,0)}}\abs{n_2'}\abs{C_{(p_1,0),\V{m}'} }^2\Big)\cr
	&\leq D_1 D_2(\Lambda_2)\mathcal{M}_B\sum_{n_1 \in \Z , j\in\{1,\ldots,N \}, s\in\{ \uparrow,\downarrow\}}\big(\norm{A\ket{(n_1,0),j,s}}^2+\norm{C^*\ket{(n_1,0),j,s}}^2\big)\cr
	&\leq D_1 D_2(\Lambda_2)\mathcal{M}_B\sum_{n_1 \in \Z , j\in\{1,\ldots,N \}, s\in\{ \uparrow,\downarrow\}}\big(\norm{A\E^{\alpha \abs{X_1}}}^2\E^{-2\alpha \abs{n_1}} +\norm{C^*\E^{\beta \abs{X_1}}}^2\E^{-2\beta \abs{n_1}}\big)<\infty,
	\end{align*}
	where $D_2(\Lambda_2)$ is a constant depending on $\Lambda_2$, we have used the hypotheses~\eqref{eqn:hyp B loc2}, and the fact that  $A$ is $\alpha$-confined in the $1^{\mathrm{st}}$-direction and $C^*$ is $\beta$-confined in the $1^{\mathrm{st}}$-direction.
	
In view of \eqref{is in ell^1} we can apply Fubini's Theorem and, by Lemma~\ref{lem:parallelogram area}, we get that the right-hand side term of \eqref{eqn:loc2 step1} reads
	\begin{align*}
	&\sum_{p_1\in\Z}\sum_{\V{m}' \in \Z^2}\sum_{\V{n}' \in \Z^2}\tr\big(A_{\V{m}', \V{n}'}B_{\V{n}', (p_1,0)}(-n_2')C_{ (p_1,0),\V{m}'} \big)\cr
	&=-\sum_{\V{m}'\in\Z^2}\tr\big({(AX_2B\chi_{2,1}C)}_{\V{m}',\V{m}'} \big).
	\end{align*}
	Observe that by hypothesis~\eqref{eqn:hyp B loc2} and Remark~\ref{rem:boundness techn}~\ref{item:Hol est}, $X_2 B\chi_{2,1}$ is bounded and thus $AX_2B\chi_{2,1}C$ is trace class, as $\chi_{2,1} C \in\TCi$ by Proposition~\ref{prop:suff traceclass cond ABC}. Therefore, one concludes that
	\begin{align*}
	\Tr(A[B,\Lambda_2] C) =  -\sum_{\V{m}'\in\Z^2}\tr\big({(AX_2B\chi_{2,1}C)}_{\V{m}',\V{m}'} \big)=-\Tr( A X_2 B\chi_{2,1}C).
	\end{align*} 
\end{proof}

\begin{lemma}
	\label{lem:loc12}
	Let $A$, $B$ and $C$ be periodic operators in $\BHi$. Let $\Lambda_1$, $\Lambda_2$ be two switch functions, respectively in the $1^{\mathrm{st}}$ and $2^{\mathrm{nd}}$-direction. 
	If $[A,\Lambda_1] B [C, \Lambda_2]$ is trace class and $A$ and $C$ satisfy 
	\begin{align}
	\label{eqn:hyp AB loc12}
    \sum_{\V{n} \in\Z^2}\abs{ A_{0,\V{n}} n_1}<\infty,\quad \sum_{\V{n} \in\Z^2}\abs{ C_{\V{m}, 0}m_2 }<\infty,
	\end{align} 
	then
	$$
	\Tr ([A,\Lambda_1] B [C, \Lambda_2]) = - \Tr (\chi_1 A X_1 B X_2 C \chi_1 ).
	$$
\end{lemma}
\begin{proof}
	Since $[A,\Lambda_1] B [C, \Lambda_2]$ is trace class, its trace can be computed through the diagonal kernel, and in view of boundedness of $[A,\Lambda_1]$, $B$ and $[C,\Lambda_2]$, one has 
	\begin{align*}
	&\Tr \Big([A,\Lambda_1] B [C, \Lambda_2] \Big) \cr 
	& =  \sum_{\V{m}\in \Z^2} \sum_{\V{n} \in \Z^2}\sum_{\V{p} \in \Z^2} \tr\big( A_{\V{m},\V{n}} (\Lambda_1(n_1) - \Lambda_1(m_1)) B_{\V{n}, \V{p}} C_{\V{p},\V{m}} (\Lambda_2(m_2) - \Lambda_1(p_2))\big).
	\end{align*}
	Performing the change of variables $\V{n}'=\V{n-m}$, $\V{p}'=\V{p-m}$ and using the periodicity, one can rewrite the right-hand side term of the last equation as
	\begin{align}
	\label{eqn:loc12 step1}
	&\sum_{\V{m}\in \Z^2} \sum_{\V{n}' \in \Z^2}\sum_{\V{p}' \in \Z^2} \tr\big( A_{\V{m},\V{n}'+\V{m}} (\Lambda_1(n_1'+m_1) - \Lambda_1(m_1))\cdot\cr
	&\phantom{=}\cdot B_{\V{n}'+\V{m}, \V{p}'+\V{m}} C_{\V{p}'+\V{m},\V{m}} (\Lambda_2(m_2) - \Lambda_1(p_2'+m_2))\big)=\cr
	&=\sum_{\V{m}\in \Z^2} \sum_{\V{n}' \in \Z^2}\sum_{\V{p}' \in \Z^2} \tr\big( A_{\V{0},\V{n}'} (\Lambda_1(n_1'+m_1) - \Lambda_1(m_1))\cdot\cr
	&\phantom{=}\cdot B_{\V{n}', \V{p}'} C_{\V{p}',\V{0}} (\Lambda_2(m_2) - \Lambda_1(p_2'+m_2))\big).
	\end{align}
	In view of the equivalence of norms on finite-dimensional vector spaces and hypothesis~\eqref{eqn:hyp AB loc12}, one has
	\begin{align*}
	&\sum_{\V{m}\in \Z^2} \sum_{\V{n}' \in \Z^2}\sum_{\V{p}' \in \Z^2} \abs{\tr\big( A_{\V{0},\V{n}'} (\Lambda_1(n_1'+m_1) - \Lambda_1(m_1)) B_{\V{n}', \V{p}'} C_{\V{p}',\V{0}} (\Lambda_2(m_2) - \Lambda_1(p_2'+m_2))\big)}\cr
	&\leq D_1 \sum_{\V{n}' \in \Z^2}\sum_{\V{p}' \in \Z^2} \abs{A_{\V{0},\V{n}'}}\sum_{m_1\in \Z} \abs{\Lambda_1(n_1'+m_1) - \Lambda_1(m_1)}\cdot\cr
	&\phantom{\leq}\cdot\abs{B_{\V{n}', \V{p}'} C_{\V{p}',\V{0}}}\sum_{m_2\in \Z}\abs{\Lambda_2(m_2) - \Lambda_1(p_2'+m_2)}\cr
	&\leq D_1 \norm{B} D_2(\Lambda_1)D_3(\Lambda_2) \sum_{\V{n}' \in \Z^2}\abs{A_{\V{0},\V{n}'}}\abs{n_1'}\sum_{\V{p}' \in \Z^2}\abs{C_{\V{p}',\V{0}}}\abs{p_2'}<\infty,
	\end{align*}
	where $D_1\in\R$, $D_2(\Lambda_1),D_3(\Lambda_2)$ are constants depending respectively on $\Lambda_1,\Lambda_2$. Therefore, applying Fubini's Theorem and Lemma~\ref{lem:parallelogram area}, one can rewrite the right-hand side term of \eqref{eqn:loc12 step1} as 
	\begin{align*}
	\sum_{\V{n}' \in \Z^2}\sum_{\V{p}' \in \Z^2} \tr\big( A_{\V{0},\V{n}'}n_1'B_{\V{n}', \V{p}'}  (-p_2')C_{\V{p}',\V{0}}\big)=-\sum_{\V{m}\in\Z^2}\tr\big({(\chi_1 AX_1BX_2 C \chi_1)}_{\V{m},\V{m}}\big)
	\end{align*}
	Observe that by hypothesis~\eqref{eqn:hyp AB loc12} and Remark~\ref{rem:boundness techn}~\ref{item:Hol est}, $\chi_1 AX_1$ and $X_2C\chi_1$ are bounded and thus $\chi_1 AX_1BX_2 C \chi_1$ is trace class. Therefore, one concludes that
	\begin{align*}
	 \Tr \Big([A,\Lambda_1] B [C, \Lambda_2] \Big) = 
	-\sum_{\V{m}\in\Z^2}\tr\big({(\chi_1 AX_1BX_2 C \chi_1)}_{\V{m},\V{m}}\big)=-\Tr(\chi_1 AX_1BX_2 C \chi_1).
	\end{align*}

\end{proof}

%%%%%  BIBLIOGRAPHY  %%%%%%%%%%%%%%%%%%%%%%%%%%%%%

%%%%%%% END MATTER

\bigskip \bigskip

{\footnotesize

\begin{tabular}{ll}
(G. Marcelli) & \textsc{Dipartimento di Matematica, ``La Sapienza'' Universit\`{a} di Roma} \\
 &  Piazzale Aldo Moro 2, 00185 Rome, Italy \\
 &  {E-mail address}: \href{mailto:marcelli@mat.uniroma1.it}{\texttt{marcelli@mat.uniroma1.it}} \\
\\
(G. Panati) & \textsc{Dipartimento di Matematica, ``La Sapienza'' Universit\`{a} di Roma} \\
 &  Piazzale Aldo Moro 2, 00185 Rome, Italy \\
 &  {E-mail address}: \href{mailto:panati@mat.uniroma1.it}{\texttt{panati@mat.uniroma1.it}} \\
\\
(C. Tauber) & \textsc{Institute for Theoretical Physics, ETH Z\"{u}rich} \\
 & Wolfgang-Pauli-Str. 27, CH-8093 Z\"{u}rich, Switzerland \\
 &  {E-mail address}: \href{mailto:tauberc@phys.ethz.ch}{\texttt{tauberc@phys.ethz.ch}} \\
\end{tabular}

} %% End footnotesize

\end{document}